\documentclass[12pt,a4paper,reqno]{amsart}
\usepackage[utf8]{inputenc}
\usepackage[T1]{fontenc}
\usepackage[top=1.5cm, bottom=1.5cm, left=2.0cm, right=2.0cm]{geometry}

\usepackage{amsthm,amsfonts,amsmath,bookmark,appendix,comment}
\usepackage{enumerate,graphicx,pdfpages,hyperref,dsfont,tabularx,bbm}
\usepackage{newtx}

\allowdisplaybreaks

\newcommand{\scp}[2]{\langle #1, #2 \rangle}

\newtheorem{theorem}{Theorem}

\newtheorem{lemma}{Lemma}[section]

\newtheorem{proposition}[lemma]{Proposition}

\newtheorem{corollary}[lemma]{Corollary}

\newtheorem{definition}[lemma]{Definition}

\newtheorem{remark}[lemma]{Remark}

\def\C {\mathbb{C}}
\def\Z {\mathbb{Z}}

\def\bq{\begin{eqnarray}}
\def\eq{\end{eqnarray}}
\def\bqq{\begin{eqnarray*}}
\def\eqq{\end{eqnarray*}}

\def\C {\mathbb{C}}

\def\gH_+ {\mathcal{Q}}

\def\C {\mathbb{C}}
\def\Z {\mathbb{Z}}

\def\vac{\Omega}

\newcommand{\spinv}{{\textbf{S}}}
\newcommand\1{{\ensuremath {\mathds 1} }}

\begin{document}

\title[Dynamics of wave packets in the Quantum Heisenberg ferromagnet]{Dynamics of low energy wave packets in the
one-dimensional \\
quantum Heisenberg ferromagnet}

\author[M. Kwaśnicki]{Mateusz Kwaśnicki}
\author[M. Napiórkowski]{Marcin Napiórkowski}
\author[S. Petrat]{S\"oren Petrat}
\author[M. Pietruszka]{Michał Pietruszka}

\address[M. Kwaśnicki]{Department of Analysis and Stochastic Processes, Wrocław University of Science and Technology, ul. Wybrzeże Wyspiańskiego 27, 50-370 Wrocław, Poland} 
\email{mateusz.kwasnicki@pwr.edu.pl}

\address[M. Napiórkowski]{ Department of Mathematical Methods in Physics, Faculty of Physics, University of
Warsaw, Pasteura 5, 02-093 Warszawa, Poland}
\email{marcin.napiorkowski@fuw.edu.pl}

\address[S. Petrat]{School of Science, Constructor University Bremen, Campus Ring 1, 28759 Bremen, Germany}
\email{spetrat@constructor.university}

\address[M. Pietruszka]{ Department of Mathematical Methods in Physics, Faculty of Physics, University of
Warsaw, Pasteura 5, 02-093 Warszawa, Poland}
\email{michalpietruszka98@gmail.com}
 
\date{\today}

\begin{abstract}
Spin wave theory plays a central role in the analysis of low-energy properties of the quantum ferromagnetic Heisenberg model. In this work, we present a first step towards a rigorous justification of this approximation in the context of quantum dynamics in one dimension. According to spin wave theory, the evolution of a spin wave state should be well approximated by its free evolution with the appropriate energy of free quasiparticles called magnons. Our main result establishes this idea rigorously for a class of spin wave packets. We focus on spin-$\frac{1}{2}$ systems, and consider wave packets composed of two magnons. We argue our result is optimal with respect to the timescales.
\end{abstract}
\maketitle

 \section{Introduction and main results}
\subsection{Introduction.} \label{sec:background}  Quantum spin chains, besides serving as paradigmatic models of quantum magnetism, constitute important platforms for quantum simulation and quantum information processing, where the propagation of localized spin excitations underlies the transport of quantum information and correlations \cite{Bose-03,Bose-07,WysChw-25}. Recent advances in ultracold atoms in optical lattices \cite{Blochetal-12}  have enabled the realization of effective Heisenberg spin chains with single-site resolution, allowing direct observation of the real-time dynamics of individual spin excitations \cite{Simon-11,Fukuhara-13a}. These developments have created a growing demand for the rigorous understanding of the dynamical behavior of interacting spin systems.

A recurring theme throughout quantum many-body physics is that interacting systems often admit effective descriptions in terms of weakly interacting quasiparticles. One of the oldest and most successful examples is spin-wave theory \cite{Blo-30,Blo-32}, which  completely neglects interactions, treating the low-energy excitations of the quantum Heisenberg ferromagnet as entirely non-interacting bosonic quasiparticles known as magnons. Within this non-interacting framework, spin flips propagate independently and freely. In reality, however, the dynamics of several magnons are genuinely interacting \cite{Bethe-31,Mattis-81,Kosevich-90}. These omitted interactions arise from two distinct physical mechanisms: the non-orthogonality of magnon states (the kinematic interaction) and the inherent coupling within the Heisenberg Hamiltonian itself \cite{Dys-56,Dys-56_2}. Their actual combined effect gives rise to nontrivial many-body phenomena, including correlation growth, modified transport, and the formation of magnon bound states \cite{Fukuhara-13b}. Consequently, because the spin-wave approximation discards these interactions entirely, its prediction of freely propagating wave packets cannot remain valid indefinitely. A fundamental question is therefore to determine the timescales over which the true interacting Schrödinger evolution remains close to this effective free dynamics.

In this work, we provide an answer to this question for a class of wave packets in the one-dimensional spin-$1/2$ quantum Heisenberg ferromagnet. We study the full quantum evolution of low-energy wave packets consisting of two magnons and compare it with the free evolution predicted by spin-wave theory. Our main result establishes explicit bounds on the norm distance between the exact interacting evolution and the corresponding non-interacting quasiparticle dynamics, identifying the precise dependence on both the system size and the evolution time. We argue the obtained timescales are optimal within the expected physical regime. Our results therefore provide a a first step towards a rigorous dynamical justification of the spin-wave approximation.

Our approach differs from the recent rigorous analyses of spin-wave theory, which were developed in the context of low free-energy asymptotics (\cite{CorGiu-12,CorGiuSei-14,CorGiuSei-15,Benedikter-2017,NapSei-19,Klippel-26}) and relied on the Holstein--Primakoff bosonization \cite{HolPri-40} of the spin system (see also \cite{ConSol-91,Toth-93} for older, probabilistic approaches). Rather than mapping the model to an approximate bosonic theory, we exploit the integrable structure of the two-magnon sector of the Heisenberg chain. This allows us to reduce the many-body dynamics to an effective one-body problem, whose spectrum and eigenfunctions can be analyzed with explicit quantitative estimates. Building on this spectral analysis, we compare the exact and free evolutions through a perturbative analysis of the corresponding propagators. To obtain the optimal timescales, we further establish uniform bounds on a class of discrete singular integral operators arising from the spectral representation. The proof therefore combines methods from integrable systems with techniques from functional and harmonic analysis, yielding a framework that is different from the bosonization-based approaches used in previous rigorous studies of spin-wave theory.

\subsection{Quantum Heisenberg model and spin waves.} \label{section:spin waves}
We consider a quantum spin-$S$ system with the Hilbert space    
\begin{align}
    \label{main_hilbert_space}
    \mathcal{H}_\Lambda := \bigotimes_{\textbf{x} \in \Lambda}\C^{2 S + 1},
\end{align}
where $\Lambda=\{1,\ldots,L\}^d \subset \Z^d$, and we additionally assume \textit{periodic boundary conditions}. Throughout this work, we adopt units in which $\hbar = 1$.
The single-site space $\mathbb{C}^{2S+1}$ is spanned by the orthonormal basis vectors conventionally written as
\begin{equation}
\vert S,m \rangle, \qquad m \in \{-S, -S+1,...,S-1, S
\}.
\label{|S,m>}
\end{equation}
 We denote by $S^{\alpha}$, $\alpha \in \{1,2,3,+,-\}$, where $S^{\pm} = S^1 \pm iS^2$, the spin operators satisfying the commutation relations
\begin{equation}
 [S^{-}, S^{+}] = -2S^{3}, \quad
[S^{3}, S^{\pm}]=\pm S^{\pm}, \quad [S^i, S^j] = i \varepsilon_{i j k} S^k, \quad i,j,k \in \{1,2,3\},
\label{CCR}
\end{equation}
where $\varepsilon_{i j k}$ is the Levi-Civita symbol. The action of $S^3$ and $S^\pm$ on the basis vectors is given by 
\begin{equation}
S^3 \vert S,m \rangle = m \vert S,m \rangle, \quad S^{\pm}\vert S,m \rangle = \sqrt{S (S+1) - m(m\pm 1)} \vert S,m \pm 1 \rangle.
\label{SpmonSm}
\end{equation}
 In particular, $|S,m\rangle$ is an eigenvector of $S^3$. Finally, we denote by $S^{\alpha}_{\textbf{x}}$ the operator acting as $S^\alpha$ on site $\textbf{x} \in \Lambda$ and as the identity elsewhere.
 
The quantum ferromagnetic Heisenberg model with nearest neighbor interactions on the lattice $\Z^d$ is defined in terms of the Hamiltonian
\begin{equation}\label{def:heisenberg_ham}
H_\Lambda :=  \sum_{\langle \textbf{x},\textbf{y} \rangle \subset \Lambda}(S^2- \spinv_{\textbf{x}} \cdot \spinv_\textbf{y}),
\end{equation}
where the sum is over all (unordered) nearest neighbor pairs $\langle \textbf{x},\textbf{y}\rangle$ of particles with spin $S$, and $ \spinv_\textbf{x}=(S_\textbf{x}^1,S_\textbf{x}^2,S_\textbf{x}^3)$ is the vector of the three spin operators acting on site $\textbf{x} \in \Lambda$. 
The Hamiltonian \eqref{def:heisenberg_ham} can be expressed in terms of the operators $S_\textbf{x}^{\pm}=S_\textbf{x}^1\pm i S_\textbf{x}^2$ and $S_\textbf{x}^3$ as
\begin{equation}
\label{def:heisenberg_ham_pm}
H_\Lambda =  \sum_{\langle \textbf{x},\textbf{y} \rangle \subset \Lambda}\left(S^2- \frac12(S_\textbf{x}^+ S_\textbf{y}^- + S_\textbf{x}^- S_\textbf{y}^+ ) - S_\textbf{x}^3 S_\textbf{y}^3 \right).
\end{equation} 
It turns out that the ground state of $H_\Lambda$ is $(2S|\Lambda|+1)$-fold degenerate (here $|\Lambda|$ denotes the number of lattice sites in $\Lambda$). Indeed, it is easy to see that one ground state is given by the "all spins down" or the "vacuum" state
\begin{displaymath}
\Omega := \bigotimes_{\mathbf{x} \in \Lambda} \vert S, -S \rangle,
\end{displaymath}
for which $H_{\Lambda}\vac = 0$ (in \eqref{def:heisenberg_ham}, we added a constant $S^2$ on each lattice bond so that the ground state energy is $0$). The name "all spins down" means that $ S_\textbf{x}^3 \vac = -S \vac\quad  \forall\,\textbf{x}\in \Lambda$, while "vacuum" corresponds to the fact that $S_\textbf{x}^- \vac =0\quad\forall\,\textbf{x}\in \Lambda$.
The other ground states can be obtained by using the operator
\begin{displaymath}    
S^+ := \sum_{\mathbf{x} \in \Lambda}S_{\mathbf{x}}^+,
\end{displaymath}
as one can easily check that $[H_{\Lambda}, S^+] = 0$, hence $(S^+)^n\vac$ is also a ground state of $H_{\Lambda}$ for $n \in \{0, 1, ..., 2S|\Lambda| \}$, with $(S^+)^0 = \1_{\mathcal{H}_{\Lambda}}$, the identity operator on $\mathcal{H}_{\Lambda}$. 

A class of excited states of \eqref{def:heisenberg_ham_pm} is constructed as follows. For a vector $\Z^d \ni \textbf{k}=(k^1,\ldots,k^d)$, define
\begin{align}
\label{S^+}
S^+ (\textbf{k})=\sum_{\textbf{x}\in \Lambda}e^{\frac{2\pi i}{L}\textbf{k}\cdot\textbf{x}} S^+_\textbf{x}.
\end{align}
This operator acts as a creation operator for spin wave excitations with momentum $\textbf{p} = 2\pi\textbf{k}/L$.
It is easy to verify that 
$$H_\Lambda S^+ (\textbf{k})\vac = S\epsilon (\textbf{k}) S^+ (\textbf{k})\vac,$$
where 
\begin{equation} \label{def:disprel}
\epsilon (\textbf{k})= 2 \sum_{j=1}^{d} \left(1- \cos \left(\tfrac{2\pi k^j}{L}\right)\right)
\end{equation}
is the dispersion relation. In particular, the state $S^+ (\textbf{k})\vac$, which will be referred to as the \textit{one-spin wave state} with wave vector $\textbf{k}$ (and momentum $\textbf{p}=2\pi\textbf{k}/L$), is an eigenstate of $H_{\Lambda}$ with energy $S\epsilon(\textbf{k})$. Analogously, a state of the form $S^+ (\textbf{k}_1)\ldots S^+ (\textbf{k}_n)\vac$, for $\textbf{k}_1,\ldots,\textbf{k}_n\in\Z^d$, will be called \textit{$n$-spin wave state}. For $n>1$, $n$-spin wave states are generally not eigenstates of $H_\Lambda$. However, for a normalized two-spin wave state $N S^+ (\textbf{k}_1) S^+ (\textbf{k}_2)\vac$ with $\textbf{k}_1 \neq \textbf{k}_2$ and $N$ a normalization constant, a direct calculation shows that
\begin{equation}\label{eq:2spinwaveenergy}
\langle N S^+ (\textbf{k}_1) S^+ (\textbf{k}_2)\vac \,\vert\, H_\Lambda  N S^+ (\textbf{k}_1) S^+ (\textbf{k}_2)\vac \rangle_{\mathcal{H}_\Lambda} = S\left(\epsilon (\textbf{k}_1)+\epsilon(\textbf{k}_2)\right) +\text{error},\qquad |\text{error}| \leq \frac{C_{\textbf{k}_1,\textbf{k}_2}}{|\Lambda|},
\end{equation}
with $C_{\textbf{k}_1,\textbf{k}_2} \leq C \big(\big(|\textbf{k}_1|/L\big)^2 + \big(|\textbf{k}_2|/L\big)^2\big)$. We notice first that the error term becomes subleading in the \textit{thermodynamic limit} $|\Lambda|\to \infty$. This leads to the idea that spin wave states are approximately eigenvectors of $H_\Lambda$ for large $\Lambda$. Secondly, the bound on $C_{\textbf{k}_1,\textbf{k}_2}$ leads us to considering only the low-energy excitation spectrum. Hence, the idea is that the energies given by sums of the form $S\sum_i \epsilon(\textbf{k}_i)$ approximate the low-energy spectrum of $H_\Lambda$. In other words, the low-energy properties of the quantum Heisenberg model can be described in terms of non-interacting spin waves.

\subsection{Main result.} 
Our goal is to rigorously verify the spin wave approximation in the context of quantum dynamics. The central object will be a wave packet of states of the form 
\begin{equation}\label{eq:initialtwomagnon}
    \Psi_{K,k}^0 =\frac{1}{L} S^+\left(\tfrac{1}{2}{\textbf{K}}+\textbf{k}\right)S^+\left(\tfrac{1}{2}{\textbf{K}}-\textbf{k}\right) \vac.
\end{equation}
This defines a two-spin-wave state with total momentum $\textbf{P} = 2\pi\textbf{K}/L$, and represents a situation when the system moves with momentum $\textbf{P}$ as a whole, and in the reference frame related to its "center of mass" (note that the position of the center of mass is ill-defined due to periodic boundary conditions), two spin waves move in opposite directions, each with momentum magnitude $|2\pi\textbf{k}/L|$. As mentioned earlier, it is not an eigenstate of the Heisenberg Hamiltonian $H_{\Lambda}$ (\ref{def:heisenberg_ham_pm}). However, we argue that under certain conditions, such states exhibit eigenstate-like behavior in the following, dynamical sense. Using the dispersion relation \eqref{def:disprel}, let us define  
\begin{equation}
    \label{E_K,p}
    E_{K,k}= S\left(\epsilon\left(\tfrac{1}{2}{\textbf{K}}+\textbf{k}\right)+\epsilon\left(\tfrac{1}{2}{\textbf{K}}-\textbf{k}\right)\right).
\end{equation} 
This formula is the expression for the non-interacting energy of $\Psi^0_{K,k}$ within spin wave approximation. If it were an eigenstate with the energy $E_{K,k}$ defined above, its time evolution would be given by the phase factor $e^{-i E_{K,k}t}$. We will refer to this as \textit{free evolution}. On the other hand, the full evolution of $\Psi^0_{K,k}$ is $e^{-i H_{\Lambda} t} \Psi_{K,k}^0$.

We show that, in the thermodynamic limit and for small momenta, these two evolutions remain close for relevant timescales. In this sense the state $ \Psi_{K,k}^0 $ behaves dynamically like an eigenstate of $H_{\Lambda}$. The theorem below provides a precise formulation of this statement. From this point onward, we consider the case  $d=1$ and $S=1/2$ only. For  simplicity,  we restrict ourselves to wave packets with fixed total momentum $P=\frac{2\pi K}{L}$  (thus fixed $K$). For further reference we define the set of relative momenta at fixed total momentum $K$ by
\begin{align} \label{def:K_K}
   \mathcal{K}_K:= 
   \begin{cases}
   \left\{ 0, 1, \dots, \frac{L-3}{2} \right\} & \text{if } L \text{ is odd, } |K| \text{ is even;} \\[4pt]
   \left\{ \frac{1}{2}, \frac{3}{2}, \dots, \frac{L-2}{2} \right\} & \text{if } L \text{ is odd, } |K| \text{ is odd;} \\[4pt]
   \left\{ 0, 1, \dots, \frac{L}{2}-1 \right\} & \text{if } L \text{ is even, } |K| \text{ is even;} \\[4pt]
   \left\{ \frac{1}{2}, \frac{3}{2}, \dots, \frac{L-3}{2} \right\} & \text{if } L \text{ is even, } |K| \text{ is odd.}
   \end{cases}
   \end{align}

\begin{theorem}\label{thm:wave-kacket}
Let $S=1/2$, $d=1$, $K\in\Z$, and let $\mathbb{N} \ni L\gg 1 $ denote the size of the one-dimensional lattice $\Lambda = \{1,...,L\}$.  For complex coefficients $(c_{K,k})_{k\in \mathcal{K}_K}$, the initial wave packet state $\Psi^0_K \in \mathcal{H}_\Lambda$ is  
\begin{align}
    \label{psi_K^0-intro}
    \Psi_K^0 &= {\sum_{k\in \mathcal{K}_K}} c_{K,k} \Psi_{K,k}^0,
\end{align}
with $\Psi_{K,k}^0$ defined in \eqref{eq:initialtwomagnon}. We choose the $(c_{K,k})_{k\in \mathcal{K}_K}$ such that $\Psi_K^0$ is normalized. For $t\in [0,+\infty)$ and the Hamiltonian $H_{\Lambda}$ given by \eqref{def:heisenberg_ham_pm}, let the real and free time evolutions of $\Psi_K^0$, denoted $\Psi_K^t$ and $\Psi_K^{t,\text{free}}$, respectively, be given by
\begin{align}
    \label{psi_K^t-intro}
    \Psi_K^t = e^{-i H_{\Lambda} t} \Psi_K^0 = {\sum_{k\in \mathcal{K}_K}} c_{K,k} e^{-i H_{\Lambda} t} \Psi_{K,k}^0 \qquad \text{and} \qquad     \Psi_K^{t,\mathrm{free}} = {\sum_{k\in \mathcal{K}_K}} c_{K,k}  e^{-i E_{K,k} t} \Psi_{K,k}^0,
\end{align}
where $E_{K,k}$ is given by \eqref{E_K,p}. For the initial data, we assume
\begin{align}
\sum_{k \in \mathcal{K}_K} \left|c_{K,k} \Gamma_{K,k}\right|^2 \leq \varepsilon^4, \qquad \text{and} \qquad 
\sum_{k \in \mathcal{K}_K \setminus \{ 0 \}} \left|c_{K,k} \frac{\Gamma_{K,k}}{\tilde E'_{K,k}}\right|^2 \leq \varepsilon^2, \label{c_assumption_1}
\end{align}
for some $0<\varepsilon<1$, and where we have defined
\begin{align}\label{def:GammaKkE'Kk}
\Gamma_{K,k} =  \cos\left( \tfrac{\pi K}{L} \right) - \cos\left( \tfrac{2\pi k}{L} \right), \qquad 
\tilde E'_{K,k} = 4\pi \cos\left(\frac{\pi K}{L}\right) \sin\left(\frac{2\pi k}{L}\right).
\end{align}
Then, for every fixed $K$ and fixed $\varepsilon>0$, there exists
$L_0(K,\varepsilon)$ such that for $L\geq L_0$ we have
\begin{equation}
\label{theorem: main equation}
\left\| \Psi_K^t  -\Psi_K^{t,\mathrm{free}} \right\|_{\mathcal{H}_\Lambda}^2 
\leq C_{K}\left(\frac{t^2 \varepsilon^4}{L^2} +  \varepsilon^2\right),
\end{equation}
with $C_{K}$ a constant independent of $L$, $\varepsilon$ and $t$.  
 \end{theorem}
The theorem is proven in Section~\ref{sec:main-kroof}. We have the following remarks concerning Theorem \ref{thm:wave-kacket}.

\begin{remark}\label{remarks_theorem}
\begin{enumerate}[(a)]
\item In the proof of Theorem \ref{thm:wave-kacket} we obtain in fact the estimate (cf. \eqref{norm_diff_final})
$$\left\| \Psi_K^t  -\Psi_K^{t,\mathrm{free}} \right\|_{\mathcal{H}_\Lambda}^2 
\leq C_{K}\left(\frac{t^2 \varepsilon^4}{L^2} + \frac{t^2}{L^5}+\varepsilon^2\right).$$
For fixed $\varepsilon$ and $L$ large enough this yields \eqref{theorem: main equation}.

\item  The free evolution defined in \eqref{psi_K^t-intro} need not preserve the norm, since the two-spin-wave states appearing in the superposition are not orthogonal. In fact, using \eqref{eq:normalization}, it is easy to see that 
$$
\|\Psi_K^{t,\mathrm{free}}\|_{\mathcal{H}_\Lambda}^2-1
=\frac{2}{L}\left(
\left|\sum_{k\in\mathcal K_K}c_{K,k}\right|^2
-
\left|\sum_{k\in\mathcal K_K}c_{K,k}e^{-iE_{K,k}t}\right|^2
\right).
$$
For the class of wave packets considered in Theorem \ref{thm:wave-kacket}, this difference is not necessarily $O(L^{-1})$. In particular, for packets involving $O(\varepsilon L)$ momentum modes, it may remain of order $\varepsilon$ as $L\to\infty$ with $\varepsilon$ fixed. One could be tempted to conclude that replacing $\Psi_K^{t,\mathrm{free}}$ with $\Psi_K^{t,\mathrm{free}}/
\|\Psi_K^{t,\mathrm{free}}\|$ might remove the $\varepsilon^2$ term from the estimate in \eqref{theorem: main equation}. This is however not the case. Even a single two-spin-wave mode has a normalized free evolution, but it is not an eigenstate of the interacting Hamiltonian. Part of the state therefore acquires a different phase under the exact evolution. This can produce a norm difference of order $\varepsilon$.

\item Conditions \eqref{c_assumption_1} are assumptions on the momentum distribution of the initial data. They are essentially encoding the fact that the wave packet has a momentum spread of order $\varepsilon$, i.e., a localization scale $\varepsilon^{-1}$ in position space. One should think of $1 \ll \varepsilon^{-1} \ll L$. Both assumptions are satisfied, e.g., if $K$ is $O(1)$, and for bounded $c_{K,k}$ with $c_{K,k} = 0$ whenever $k/L > \varepsilon$, see Remark~\ref{rem_end_of_proof}.  
    \item The time-independent term on the right-hand side of \eqref{theorem: main equation} is a remainder that is nonphysical for small times $t$. Thus, the result obtained in the theorem is not optimal for very small timescales and should be considered only for large times. For short times a Duhamel type bound yields an upper bound that goes to zero as $t\to 0$  (cf. Proposition \ref{prop:shorttimes}).
    \item The wave packet momentum $\varepsilon$ and the box size $L$ are two independent parameters. The estimate in  \eqref{theorem: main equation} supports the claim that the spin wave approximation should be more accurate in the thermodynamic limit (large $L$) and in the small energy sector (small $\varepsilon$).
\end{enumerate}
\end{remark}

\subsection{Optimality heuristics.}\label{sec:heuristics} We believe the time-dependent term on the right-hand side of \eqref{theorem: main equation} is optimal. We can furthermore present heuristic evidence that the time-independent term is optimal as well. For the time-dependent term, let us present a heuristic argument which holds in any dimension and for any $S$. Let
\begin{equation}\label{def_Psi_heuristic_argument}
\Psi^{(\textbf{k})} = N  S^+(\textbf{k})S^+(-\textbf{k})\vac
\end{equation}
be a two-spin wave state, where $N$ is the normalization constant and $\textbf{k}\in\Z^d$. In the setting of Theorem~\ref{thm:wave-kacket}, this state satisfies the assumptions \eqref{c_assumption_1} if $k = c \varepsilon L$ for sufficiently small fixed $c$. We expect that
\begin{displaymath}
E^{(\textbf{k})} = 2S\epsilon(\textbf{k})= 4 S\sum_{j=1}^d \left( 1- \cos\left(\tfrac{2\pi k^j}{L}\right)\right)
\end{displaymath}
is a good approximation for the energy of $\Psi^{(\textbf{k})}$. For simplicity of notation, let us present the argument here for the case where there is no spin wave degeneracy, i.e., $E$ is injective as a function of $\textbf{k}\in(\Z/L)^d$ (up to $E^{(\textbf{k})} = E^{(-\textbf{k})}$), which is true only in one dimension. In higher dimension one needs to take appropriate sums of spin waves with degenerate energy $E^{(\textbf{k})}$ as initial states. For small enough $\textbf{k}$, we expect that $E^{(\textbf{k})}$ is very close to exactly one eigenvalue $\lambda_{\textbf{k}}$ of $H_\Lambda$ (restricted to total momentum $\textbf{K}=0$) with eigenstate $\widetilde{\Psi}_{\textbf{k}}$. Hence, we decompose
\begin{equation}\label{def:Psi}
\Psi^{(\textbf{k})} = \sqrt{1-\delta_{\textbf{k}}^2}\,\widetilde{\Psi}_{\textbf{k}} + \delta_{\textbf{k}}\,\widetilde{\Psi}_{\textbf{k},\mathrm{rest}}
\end{equation}
for some $0 < \delta_{\textbf{k}} \leq 1$, where we choose $\widetilde{\Psi}_{\textbf{k},\mathrm{rest}}$ orthogonal to $\widetilde{\Psi}_{\textbf{k}}$ and both normalized. Due to the orthogonality, we find that
\begin{align}
\label{est_pre}
\begin{split}
&\left\| e^{-i H_{\Lambda} t}\Psi^{(\textbf{k})}  - e^{-i E^{(\textbf{k})} t}\Psi^{(\textbf{k})} \right\|_{\mathcal H_{\Lambda}}^2 \\
&\qquad= 2\Re \left( 1 - \langle\Psi^{(\textbf{k})}\, \vert \,e^{-i(H_{\Lambda}-E^{(\textbf{k})})t}\Psi^{(\textbf{k})} \rangle_{\mathcal H_{\Lambda}} \right) \\
&\qquad= 2\Re \left( 1 - (1-\delta_{\textbf{k}}^2)\langle \widetilde{\Psi}_{\textbf{k}}  \,\vert\,e^{-i(\lambda_{\textbf{k}}-E^{(\textbf{k})})t}\widetilde{\Psi}_{\textbf{k}} \rangle_{\mathcal H_{\Lambda}}\right) -2 \delta_{\textbf{k}}^2 \Re\,  \langle \widetilde{\Psi}_{\textbf{k},\mathrm{rest}}\,\vert\, e^{-i(H_{\Lambda}-E^{(\textbf{k})})t} \widetilde{\Psi}_{\textbf{k},\mathrm{rest}}\rangle_{\mathcal H_{\Lambda}} \\ 
&\qquad= 2\left( 1 - (1-\delta_{\textbf{k}}^2) \cos\left( (\lambda_{\textbf{k}} - E^{(\textbf{k})}) t\right) - \delta_{\textbf{k}}^2 \Re\,  \langle \widetilde{\Psi}_{\textbf{k},\mathrm{rest}}\,\vert\, e^{-i(H_{\Lambda}-E^{(\textbf{k})})t} \widetilde{\Psi}_{\textbf{k},\mathrm{rest}}\rangle_{\mathcal H_{\Lambda}} \right),
\end{split}
\end{align}
which, by unitarity of $e^{-iH_{\Lambda}t}$, yields the lower and upper bounds
\begin{align}
\label{est}
\begin{split}
4 (1-\delta_{\textbf{k}}^2)\sin^2\left( \frac{(\lambda_{\textbf{k}} - E^{(\textbf{k})})}{2} t\right) &\leq \left\| e^{-i H_{\Lambda} t}\Psi^{(\textbf{k})}  - e^{-i E^{(\textbf{k})} t}\Psi^{(\textbf{k})} \right\|_{\mathcal H_{\Lambda}}^2 \leq (\lambda_{\textbf{k}}-E^{(\textbf{k})})^2 t^2 + 4\delta_{\textbf{k}}^2.
\end{split}
\end{align}
Note that the time-dependent terms are entirely due to the phase mismatch. From the definition of $\Psi^{(\textbf{k})}$ in \eqref{def_Psi_heuristic_argument} one can show by direct computation that
\begin{equation}
\label{first_moment}
\begin{split}
\langle\Psi^{(\textbf{k})}\,\vert\, \big( H_{\Lambda}-E^{(\textbf{k})} \big) \Psi^{(\textbf{k})}\rangle_{\mathcal H_{\Lambda}}=(1-\delta_{\textbf{k}}^2) (\lambda_{\textbf{k}}-E^{(\textbf{k})}) + \delta_{\textbf{k}}^2 \langle \widetilde{\Psi}_{\textbf{k},\mathrm{rest}}\,\vert\,(H_{\Lambda}-E^{(\textbf{k})})\widetilde{\Psi}_{\textbf{k},\mathrm{rest}}\rangle_{\mathcal H_{\Lambda}} \approx |\Lambda|^{-1} \left(\frac{|\textbf{k}|}{L}\right)^2 
\end{split}
\end{equation} 
uniformly in $S$, as in \eqref{eq:2spinwaveenergy}. As there seems to be no reason for any cancellation between the two summands in the middle term of \eqref{first_moment}, we guess that
\begin{equation}\label{Delta_E}
\lambda_{\textbf{k}}-E^{(\textbf{k})} \approx |\Lambda|^{-1} \left(\frac{|\textbf{k}|}{L}\right)^2
\end{equation}
which inserted into \eqref{est} gives (in the 1D case) exactly the timescales as in Theorem~\ref{thm:wave-kacket}. Note that this argument is uniform in $S$, and works in any dimension $d$ (when taking into account obvious degeneracies appropriately). Furthermore, note that the asymptotic behavior \eqref{Delta_E} can be computed exactly for one dimension and spin $S=\frac{1}{2}$, see Equation~\eqref{exact_energy_diff_one_dim_K_zero} in Remark~\ref{remark_asymptotic_bound}.

An estimate on the error $\delta_{\textbf{k}}$ in \eqref{est} is more involved and depends on the dimension. Let us consider the setting of our theorem, namely one dimension and $S=\frac{1}{2}$. Then for $K = 0$ there are explicit Bethe ansatz solutions for the spectrum and eigenfunctions of $H_\Lambda$. With this, one can compute an exact expression for $\delta_k$ and show that
\begin{equation}\label{delta_approx_heuristic_argument}
\delta_k^2 \approx \frac{\pi^2k^2}{3 L^2},
\end{equation}
which again coincides with the time-independent term of our bound \eqref{theorem: main equation} in Theorem~\ref{thm:wave-kacket}. We prove the asymptotics \eqref{delta_approx_heuristic_argument} in Remark~\ref{remark_asymptotic_bound}.

\subsection{Organization of the article and strategy of proof.}  As mentioned in the introduction, unlike in the recent works on the free energy asymptotics in which the analysis was based on the Holstein–Primakoff representation, the approach developed here is  different. Rather than bosonizing the model, we exploit the structure of the two-magnon sector of the spin chain and proceed in the spirit of integrability. 

More precisely, as shown in Section \ref{section:one-body-reduction}, the dynamics of two spin waves can be reduced to a one-body problem with an effective Hamiltonian. In order to analyze this effective dynamics, we establish the asymptotic form of the eigenvalues and eigenfunctions of this operator. This is done in Section \ref{section:spectral problem}.  

The proof of the main theorem is then given in Section \ref{chapter:proof}. To reach the sharp timescales claimed in our main theorem one must control terms that can be expressed via the action of certain operators with (discrete) singular kernels. In particular, we need to show that the norms of these operators are bounded uniformly in the system size. This is done in Appendix \ref{appendix:boundedness}. \newline

\noindent {\bf{Notation.}} 
The symbol $\1$, with an appropriate subscript, is used in two contexts: to denote the identity operator in a vector space or the characteristic function of a set. Instead of $\1_{\{x\;\vert\;\text{condition for }x\}}(x)$, we often write $\1_{\text{condition for }x}$ to indicate that a certain term appears conditionally. For example, $\1_{L\;\text{even}}f(L)$ means that the term $f(L)$ is included only when the parameter $L$ is even. 
To improve the readability of the calculations, we frequently use the same symbol, typically $C$, to denote generic constants, which are independent of the problem's parameters, particularly the lattice size $L$. \newline

\noindent {\bf{Acknowledgments.}} M.~N. and S.~P. are grateful to Robert Seiringer for suggesting the problem and for many inspiring discussions. The work of M.~N. was supported by the National Science Centre (NCN) grant Sonata Bis 13 number 2023/50/E/ST1/00439. M.~K. was supported by the National Science Centre, Poland, grant no. 2023/49/B/ST1/04303. \newline

\noindent {\bf{AI disclosure.}}  Language models GPT-5.6 Sol, GPT-Astra and Gemini 3.1 Pro have been used to improve the
presentation and in proofreading. All ideas, mathematical formulations, statements of theorems, and proofs were conceived and
developed entirely by the authors.

\section{Reduction to the one-body problem}\label{section:one-body-reduction}
In this section, we set the stage for the proof of Theorem~\ref{thm:wave-kacket}. We demonstrate that the time evolution of the two-spin wave state can be described in terms of the time evolution of its wave function and we reduce it to the one-body dynamics.

\subsection{Time evolution equation for the wave function}
\label{section:time evolution for wave function}
Let us briefly consider the general state $\Psi^t \in \mathcal{H}_{\Lambda}$ with two excitations, which can be expressed as
\begin{equation}
\label{eq:wave_fct}
\Psi^t= \sum_{x,y \in \Lambda} \psi^t (x,y) \, S_x^+ S_y^+ \vac,
\end{equation}
where $\psi^t$ is a time-dependent, complex-valued wave function. We impose periodic boundary conditions in both variables: 
\begin{align}
    \label{PBC}
    \psi^t(x+L,y)=\psi^t(x,y), \qquad \psi^t(x,y+L)=\psi^t(x,y) \qquad \forall\,x,y\in\Lambda.
\end{align}
 Note that $S_x^+ S_y^+ \vac = S_y^+ S_x^+ \vac$, and since $S_x^+ S_x^+ \vac=0$ when $S=1/2$, the values $\psi^t(x,x)$ are irrelevant. Therefore, it is convenient to put
\begin{align}
    \label{psi_conditions}
        \psi^t(x,x)=0, \qquad
        \psi^t(x,y)=\psi^t(y,x) \qquad \forall\,x,y\in\Lambda.
\end{align}
Let $\mathfrak{h}_{\Lambda}$ be the space of complex-valued wave functions satisfying \eqref{PBC} and \eqref{psi_conditions}. Then the $\mathcal{H}_\Lambda$-norm of $\Psi^t$ reads
\begin{equation}
\label{eq:spinscalarproduct}
\|\Psi^t\|_{\mathcal{H}_{\Lambda}}^2 = \langle\Psi^t \left. | \right. \Psi^t\rangle_{\mathcal{H}_{\Lambda}} =2 \langle \psi^t \left. | \right. \psi^t\rangle_{\mathfrak{h}_{\Lambda}} := 2 \sum_{\substack{x,y\in\Lambda\\ x\neq y}} |\psi^t(x,y)|^2,
\end{equation}
where the last equality introduces the inner product on $\mathfrak{h}_{\Lambda}$. Using \eqref{eq:wave_fct}, a direct computation shows that the Schr\"odinger equation $i\partial_t \Psi^t = H_\Lambda \Psi^t$
can be reformulated in terms of the wave function $\psi^t$ as
\begin{equation}
\label{eq:Schrodinger_wave_fct}
H_\Lambda \Psi^t = \sum_{x,y \in \Lambda} \Big[ i\partial_t \psi^t(x,y) \Big] S_x^+ S_y^+ \vac = \sum_{x,y \in \Lambda} \Big[ (h_\Lambda \psi^t)(x,y) \Big] S_x^+ S_y^+ \vac, 
\end{equation}  
where the operator $h_{\Lambda}:\mathfrak{h}_{\Lambda}\to\mathfrak{h}_{\Lambda}$ acts as 
\begin{equation}
\label{eq:wave_fct_eq}
 (h_{\Lambda}\psi^t)(x,y) = \begin{cases}
 \frac12\left(-\Delta^{G}_1 \psi^t \right)(x,y) +  \frac12\left(-\Delta^{G}_2 \psi^t \right)(x,y)
-\psi^t(x,y) & \quad x\sim y,\\
 \frac12\left(-\Delta^{G}_1 \psi^t\right)(x,y) +  \frac12\left(-\Delta^{G}_2 \psi^t\right)(x,y) & \quad x\not \sim y, \, x \ne y, \\ 0 & \quad x = y, \end{cases}
\end{equation}
where $x\sim y$ indicates that $x,y\in\Lambda$ are nearest neighbors. Here, we denote by $-\Delta^{G}_1$, $-\Delta^{G}_2$ the discrete graph Laplacians acting on the first and second arguments, i.e., 
\begin{equation} \label{def:graphlaplacian}
\begin{split}
\left(-\Delta^{G}_1 \psi^t\right)(x,y)&=\sum_{\substack{z\in\Lambda \\ z\sim x}}\left(\psi^t(x,y)-\psi^t(z,y) \right),\\
\left(-\Delta^{G}_2 \psi^t\right)(x,y)&=\sum_{\substack{z\in\Lambda \\ z\sim y}}\left(\psi^t(x,y)-\psi^t(x,z) \right).
\end{split}
\end{equation}
Note that \eqref{eq:Schrodinger_wave_fct} implies
\begin{equation}
    \label{eq:Schrodinger_psi_small}
    \psi^t(x,y) = \left(e^{-i h_{\Lambda}t} \psi^0 \right)(x,y),
\end{equation}
which, together with \eqref{eq:wave_fct}, allow us to analyze the time evolution of two-spin excitations in terms of the Hamiltonian $h_{\Lambda}$, which acts on wave functions. We now apply this general result to the specific wave packet introduced earlier.

\subsection{The one-body operator}\label{sec_onebody_discussion_time_evolution} 
Throughout this section, we assume that $K\in\Z$.
The initial state $\Psi_K^0$, given by \eqref{psi_K^0-intro}, can be rewritten as
\begin{align}
\label{initial-step}
 \Psi_K^0 &= \sum_{k\in\mathcal{K}_K} c_{K,k} \frac{1}{L}   S^+\left(\tfrac{K}{2}+k\right)S^+\left(\tfrac{K}{2}-k\right) \vac \nonumber\\
    &= \frac{1}{L} \sum_{k\in\mathcal{K}_K} c_{K,k}\sum_{\substack{x,y\in\Lambda \\ x\neq y}} e^{\frac{2\pi i}{L}\left(\frac{K}{2}+k\right)x} e^{\frac{2\pi i}{L}\left(\frac{K}{2}-k\right)y} S^+_x S^+_y \vac \nonumber\\
    &= \sum_{\substack{x,y\in\Lambda \\ x\neq y}} e^{\frac{2\pi i}{L} K\left(\frac{x+y}{2}\right)} \frac{1}{L} \sum_{k\in\mathcal{K}_K} c_{K,k}\cos\left( \tfrac{2\pi k}{L}  (x-y) \right) S^+_x S^+_y \vac \nonumber \\
    &= \sum_{\substack{x,y\in\Lambda \\ x\neq y}} \psi_K^0(x,y) S^+_x S^+_y \vac,
\end{align}
where the summation is over $x\neq y$ because $S_x^+ S_x^+ \Omega = 0$, and the diagonal elements of the wave function, $\psi^0_K(x,x)$, are irrelevant. The appearance of the cosine function is a consequence of the invariance of the summation  under exchanging $x$ and $y$ and the commutativity of $S^+_x$ and  $S^+_y$. As in the previous section, we put $\psi^0_K(x,x) = 0$ $\forall\, x\in\Lambda$, while
\begin{equation}
    \label{special_case_psi_K^0}
    \psi_K^0(x,y) = e^{\frac{2\pi i}{L}K\left(\frac{x+y}{2}\right)} \frac{1}{L} \sum_{k\in\mathcal{K}_K} c_{K,k}  \cos\left(\tfrac{2\pi k}{L}  (x-y) \right)\qquad\forall\,x,y\in\Lambda,\,x \neq y.
\end{equation}
This wave function fulfills \eqref{PBC} and \eqref{psi_conditions}, and its time evolution is given by \eqref{eq:Schrodinger_psi_small}. Note that the operator $h_{\Lambda}$, the generator of this evolution, is a two-body operator. However, for the wave functions of the form \eqref{special_case_psi_K^0}, the problem can be reduced to the analysis of a certain one-body operator as follows. Let $\phi^0_K$ be the general wave function of the form
\begin{equation}
\label{eq:ansatz}
\phi^0_K(x,y) = e^{\frac{2\pi i}{L}K\left(\frac{x+y}{2}\right)}f_K(x-y)\qquad x,y\in\Lambda,
\end{equation}
for some one-body function $f_K$ describing the relative motion. To ensure the periodic boundary conditions \eqref{PBC} and symmetries \eqref{psi_conditions} on $\phi^0_K$, we must have 
\begin{equation}
\label{eq:symmetries_of_f}
f_K(-X) = f_K(X), \qquad f_K(X+L) = (-1)^Kf_K(X), \qquad f_K(0) = 0.
\end{equation}
Of course, $\psi_K^0$ from \eqref{special_case_psi_K^0} is a special case of \eqref{eq:ansatz}. Note that if $X = L/2 + R$ for some $R$, then (\ref{eq:symmetries_of_f}) gives 
\begin{equation} \label{symmetryR}
    f_K\left(\tfrac{L}{2}+R\right) = f_K\left(-\tfrac{L}{2}-R\right) = (-1)^K f_K\left(-\tfrac{L}{2}-R+L\right) = (-1)^K f_K\left(\tfrac{L}{2}-R\right).
\end{equation} 
This justifies the following definition of the one-body Hilbert space. 
\begin{definition}
    For a set
    \begin{displaymath}
     \mathcal{R}_L := 
\begin{cases} 
\big\{\tfrac{1}{2},\tfrac{3}{2},\ldots,\tfrac{L}{2}-1 \big\} &  \text{$L$ odd;} \\[10pt]
\big\{0,1,\ldots,\tfrac{L}{2}-1 \big\} &  \text{$L$ even,}
\end{cases}
\end{displaymath}
we define the one-body Hilbert space $\mathcal{H}_{K,L}^{(1)}$  as
\begin{align}
   \label{H_1_def}
        \mathcal{H}_{K,L}^{(1)} := \Big\{ f_K: \{1,\ldots,L-1\} \to \C\left.  \Big| \right.  
        f_K\left(\tfrac{L}{2}+R\right) = (-1)^K f_K\left(\tfrac{L}{2}-R\right) ~\forall R \in \mathcal{R}_L \Big\}.
\end{align}
The inner product on $\mathcal{H}_{K,L}^{(1)}$ reads
\begin{displaymath}
    \langle f_K \left| \right. g_K \rangle_{\mathcal{H}_{K,L}^{(1)}} := \sum_{X=1}^{L-1} \overline{f_K(X)}g_K(X)\qquad \forall\,f_K, g_K \in \mathcal{H}^{(1)}_{K,L}.
\end{displaymath}
\end{definition}
\begin{remark}
\begin{enumerate}[(a)]
   \item The domain $\{1,\ldots,L-1\}$ is sufficient, as one can use \eqref{eq:symmetries_of_f} to extend $f_K \in \mathcal{H}_{K,L}^{(1)}$ to all of $\mathbb{Z}$. Conversely, every complex-valued function on $\Z$ satisfying \eqref{eq:symmetries_of_f} can be regarded as an element of $\mathcal{H}_{K,L}^{(1)}$.

\item The symmetry condition \eqref{symmetryR} uniquely determines $f_K(X)$ for $X > L/2$ from its values for $X < L/2$:
   \begin{itemize}
       \item For $L$ odd, the midpoint $L/2$ is non-integer. The values on $\big\{1, \ldots, \frac{L-1}{2}\big\}$ freely determine the remaining values on $\big\{\frac{L+1}{2}, \ldots, L-1\big\}$, yielding $\dim \mathcal{H}_{K,L}^{(1)} = \frac{L-1}{2}$.
       \item For $L$ even, $X = L/2$ is a lattice site (corresponding to $R=0$), giving $f_K(L/2) = (-1)^K f_K(L/2)$. Thus, $f_K(L/2)$ is a free degree of freedom if $K$ is even, but vanishes if $K$ is odd. The values on $\big\{1, \ldots, \frac{L}{2}-1\big\}$ freely determine those on $\big\{\frac{L}{2}+1, \ldots, L-1\big\}$.
   \end{itemize}
   Consequently, the dimension of the one-body space is given by
   \begin{align}
   \label{appendix_J_def}
   \dim \mathcal{H}_{K,L}^{(1)} = 
   \begin{cases} 
   \frac{L-1}{2} & \text{if } L \text{ is odd;} \\[4pt]
   \frac{L}{2} & \text{if } L \text{ is even, } |K| \text{ is even;} \\[4pt]
   \frac{L}{2} - 1 & \text{if } L \text{ is even, } |K| \text{ is odd.}
   \end{cases}
   \end{align}

   \item Quantization of relative momentum: the relative motion modes $f_{K}(X) = \cos\left(\frac{2\pi k}{L}X\right)$ satisfy the pseudo-periodicity condition $f_K(X+L) = (-1)^K f_K(X)$ in \eqref{eq:symmetries_of_f} if and only if $\cos(2\pi k) = (-1)^K$. This dictates the quantization of $k$:
   \begin{itemize}
       \item For $|K|$ even, $k$ takes non-negative integer values: $k \in \{0, 1, 2, \dots\}$.
       \item For $|K|$ odd, $k$ takes positive half-integer values: $k \in \big\{\frac{1}{2}, \frac{3}{2}, \frac{5}{2}, \dots\big\}$.
   \end{itemize}
   To avoid linear dependence of those relative motion modes on the finite lattice $\{1,\dots,L-1\}$, $k$ is bounded above such that the total number of modes equals $\dim \mathcal{H}_{K,L}^{(1)}$. The full set of independent relative momentum modes is given by
   \begin{align}
   \mathcal{K}_K:= 
   \begin{cases}
   \left\{ 0, 1, \dots, \frac{L-3}{2} \right\} & \text{if } L \text{ is odd, } |K| \text{ is even;} \\[4pt]
   \left\{ \frac{1}{2}, \frac{3}{2}, \dots, \frac{L-2}{2} \right\} & \text{if } L \text{ is odd, } |K| \text{ is odd;} \\[4pt]
   \left\{ 0, 1, \dots, \frac{L}{2}-1 \right\} & \text{if } L \text{ is even, } |K| \text{ is even;} \\[4pt]
   \left\{ \frac{1}{2}, \frac{3}{2}, \dots, \frac{L-3}{2} \right\} & \text{if } L \text{ is even, } |K| \text{ is odd.}
   \end{cases}
   \end{align}
This corresponds to \eqref{def:K_K}.
   \end{enumerate}
\end{remark}

The following (well-known, e.g., \cite{wortis}) result will be used in our analysis. 
\begin{proposition} For $\phi^0_K$ of the form \eqref{eq:ansatz}, where $f_K$ is treated as an element of $\mathcal{H}_{K,L}^{(1)}$, we have
    \begin{align}
    \label{h_lambda_K_def}
    (h_{\Lambda}\phi^0_K)(x,y) = e^{\frac{2\pi i}{L}K\left(
\frac{x+y}{2}\right)} (h_{\Lambda,K}f_K)(x-y),
\end{align}
where the one-body operator $h_{\Lambda,K} : \mathcal{H}_{K,L}^{(1)} \to \mathcal{H}_{K,L}^{(1)}$ acts as
\begin{equation}
\label{eq_onebody_Ham}
(h_{\Lambda,K} f_K)(X) =
\begin{cases}
 f_K(1)-\cos\left(\tfrac{\pi K}{L} \right) f_K(2) &\quad X=1;\\
 2 f_K(X)-\cos\left(\tfrac{\pi K}{L} \right) \Big( f_K(X+1) + f_K(X-1) \Big) &\quad X \in \{2,\ldots,L-2\};\\
f_K(L-1)-\cos\left(\tfrac{\pi K}{L} \right) f_K(L-2) &\quad X = L-1.
\end{cases}
\end{equation} 
\end{proposition}
\begin{proof} Straightforward by using \eqref{eq:wave_fct_eq} and \eqref{def:graphlaplacian}, and then applying \eqref{eq:symmetries_of_f} to $f_K$. 
\end{proof}
\noindent This reduction allows us to analyze the dynamics of the original two-spin wave state by studying the spectral properties of a simpler, one-body operator $h_{\Lambda,K}$.
In fact, from the expansion of \eqref{h_lambda_K_def} into a power series, one immediately obtains
\begin{corollary} The time evolution of $\phi^0_K$ reads
    \begin{align}
    \label{e^h_lambda,P}
    \phi^t_K(x,y) = (e^{-i h_{\Lambda} t} \phi^0_K)(x,y) = e^{\frac{2\pi i}{L}K\left(
\frac{x+y}{2}\right)} (e^{-i h_{\Lambda,K} t} f_K)(x-y).
\end{align}
\end{corollary}
\noindent Of course, the above formula applies to \eqref{special_case_psi_K^0}. We see that the full information about the time evolution is therefore encoded in the one-body operator $h_{\Lambda,K}$. In the next section, we proceed with a deeper analysis of this operator.

\section{Spectral analysis of the effective Hamiltonian}
\label{section:spectral problem}
 In this section we analyze the spectral properties of the effective one-body Hamiltonian. We begin by introducing an operator that will be used to rewrite the one-body Hamiltonian $h_{\Lambda,K}$. We assume that $K$ and $L$ are such that $\cos(\pi K /L) \neq 0$. Otherwise, the problem becomes trivial, since according to \eqref{eq_onebody_Ham}, $h_{\Lambda,K}$ is then a diagonal operator with $\lambda_{K,1} = 1$ and $\lambda_{K,j} = 2 \quad \forall j\neq 1$, and it is easy to verify that the normalized eigenvectors are given by $f_{K,j}(X) = (\delta_{X,j}+(-1)^K\delta_{X,L-j})/\sqrt{2}$ (for $j \neq L/2$). For this section let us slightly change the notation in \eqref{eq_onebody_Ham}.
\begin{definition} Let $L \in \mathbb{N}$ with $L \ge 4$ and $K \in \mathbb{Z}$, $\lvert K \rvert < L$. Define the coupling parameter
\begin{equation}
\delta := \cos\left(\frac{\pi K}{L}\right) \in (0, 1].
\end{equation}
The relative one-body Hamiltonian is the linear operator $h_\delta: \ell^2(\{1, \dots, L-1\}) \to \ell^2(\{1, \dots, L-1\})$ defined by
\begin{equation}
\label{eq:h_definition}
(h_\delta f)(X) = \begin{cases}
f(1) - \delta f(2) & \text{if } X = 1, \\[4pt]
2 f(X) - \delta f(X - 1) - \delta f(X + 1) & \text{if } 2 \le X \le L - 2, \\[4pt]
f(L - 1) - \delta f(L - 2) & \text{if } X = L - 1.
\end{cases}
\end{equation}
\end{definition}
We have the following spectral theorem for the operator $h_\delta$, essentially rephrasing Bethe's original work \cite{Bethe-31}.

\begin{theorem}[Spectral Decomposition of $h_\delta$]
\label{thm:spectral_decomposition}
The operator $h_\delta$ is real-symmetric and possesses $L-1$ simple eigenvalues $\lambda_0 < \lambda_1 < \dots < \lambda_{L-2}$ with mutually orthogonal eigenvectors $\{f_0, f_1, \dots, f_{L-2}\}$, categorized as follows:
\begin{enumerate}
    \item The $\cosh$ solution: $f_0(X) = \cosh\left(\beta_0 \left(X - \frac{L}{2}\right)\right)$ with $\beta_0 \in [0,\operatorname{arcosh}(\tfrac{1}{\delta})]$ and eigenvalue
    \begin{equation}
    \lambda_0 = 2 - 2 \delta \cosh(\beta_0).
    \end{equation}
    
    \item The $\sinh$ solution: $f_1(X) = \beta_1^{-1} \sinh\left(\beta_1 \left(X - \frac{L}{2}\right)\right)$ if $\delta \le 1 - \frac{2}{L}$,  with $\beta_1 \in [0,\log(1/\delta)]$ and eigenvalue
        \begin{equation}
        \lambda_1 = 2 - 2 \delta \cosh(\beta_1).
        \end{equation}
        The case $\beta_1 = 0$ is included here by continuity; then $f_1(X) = X - \tfrac{L}{2}$.
        \item Regular solutions: 
      $f_n(X) = \cos\left(\frac{n \pi}{2} + \alpha_n \left(X - \frac{L}{2}\right)\right)$ with  $\alpha_n\in[\frac{(n-1)\pi}{L-2},\frac{n\pi}{L-1}]$  and eigenvalues
    \begin{equation}
    \lambda_n = 2 - 2 \delta \cos(\alpha_n)
    \end{equation}
with $1\leq n\leq L-2$ if $\delta>1-2/L$ or $2\leq n\leq L-2$ (otherwise).
\end{enumerate}
\end{theorem}

\begin{proof} The fact that all listed functions are eigenfunctions with the corresponding eigenvalues is a straightforward calculation. To prove the theorem we thus need to establish the existence of such $\beta_0,\beta_1,\alpha_n$ that the corresponding conditions are met.\\
\noindent \textit{(1)}  Let us first inspect when a $\cosh$ solution $f_0$ exists. The boundary condition reads
\[
 (2 - 2 \delta \cosh(\beta_0)) f_0(1) = \lambda_0 f_0(1) = h f_0(1) = f_0(1) - \delta f_0(2) ,
\]
which simplifies to
\[
 f_0(1) = \delta (2 \cosh(\beta_0) f_0(1) - f_0(2)) .
\]
Using the definition of $f_0$, we obtain
\[
 \cosh(\tfrac{L - 2}{2} \beta_0) = \delta (2 \cosh(\beta_0) \cosh(\tfrac{L - 2}{2} \beta_0) - \cosh(\tfrac{L - 4}{2} \beta_0)) ,
\]
which, by $\cosh(a + b) + \cosh(a - b) = 2 \cosh a \cosh b$, simplifies to
\[
 \cosh(\tfrac{L - 2}{2} \beta_0) = \delta \cosh(\tfrac{L}{2} \beta_0) .
\]
The function $\beta \mapsto \cosh(\tfrac{L - 2}{2} \beta) / \cosh(\tfrac{L}{2} \beta)$ is strictly decreasing on $[0, \infty)$ and takes all values in $(0, 1]$. Hence, the $\cosh$ solution always exists.

Since
$$\cosh(\tfrac{L - 2}{2} \beta) / \cosh(\tfrac{L}{2} \beta) \le 1 / \cosh(\beta),$$ we have $\beta_0 \le \operatorname{arcosh}(\tfrac{1}{\delta})$. This gives the bound from the theorem.\\
\noindent \textit{(2)} We repeat the above reasoning for a $\sinh$ solution $f_1$. The boundary condition is
\[
 (2 - 2 \delta \cosh(\beta_1)) f_1(1) = \lambda_1 f_1(1) = h f_1(1) = f_1(1) - \delta f_1(2) ,
\]
that is,
\[
 f_1(1) = \delta (2 \cosh(\beta_1) f_1(1) - f_1(2)) .
\]
By the definition of $f_1$, we get
\[
 \sinh(\tfrac{L - 2}{2} \beta_1) = \delta (2 \cosh(\beta_1) \sinh(\tfrac{L - 2}{2} \beta_1) - \sinh(\tfrac{L - 4}{2} \beta_1)) ,
\]
which, by $\sinh(a + b) + \sinh(a - b) = 2 \cosh b \sinh a$, simplifies to
\[
 \sinh(\tfrac{L - 2}{2} \beta_1) = \delta \sinh(\tfrac{L}{2} \beta_1) .
\]
The function $\beta \mapsto \sinh(\tfrac{L - 2}{2} \beta) / \sinh(\tfrac{L}{2} \beta)$ is strictly decreasing on $[0, \infty)$ and takes all values in $(0, 1 - \tfrac{2}{L}]$. Hence, the $\sinh$ solution exists if and only if $\delta \le 1 - \tfrac{2}{L}$. We have 
$$\sinh(\tfrac{L - 2}{2} \beta) / \sinh(\tfrac{L}{2} \beta) \le e^{-\beta},$$ and so $\beta_1 \le \log \tfrac{1}{\delta}$. This then leads to the bound from the theorem.\\
\textit{(3)} A similar calculation allows us to describe the regular solutions. Again we begin with the boundary condition
\[
 (2 - 2 \delta \cos(\alpha_n)) f_n(1) = \lambda_n f_n(1) = h f_n(1) = f_n(1) - \delta f_n(2) ,
\]
that is,
\[
 f_n(1) = \delta (2 \cos(\alpha_n) f_n(1) - f_n(2)) .
\]
By the definition of $f_n$, we get
\[
 \cos(\tfrac{L - 2}{2} \alpha_n - \tfrac{n \pi}{2}) = \delta (2 \cos(\alpha_n) \cos(\tfrac{L - 2}{2} \alpha_n - \tfrac{n \pi}{2}) - \cos(\tfrac{L - 4}{2} \alpha_n - \tfrac{n \pi}{2})) .
\]
Using $\cos(a + b) + \cos(a - b) = 2 \cos b \cos a$, we simplify this condition to
\[
 \cos(\tfrac{L - 2}{2} \alpha_n - \tfrac{n \pi}{2}) = \delta \cos(\tfrac{L}{2} \alpha_n - \tfrac{n \pi}{2}) .
\]
We claim that there is a solution $f_n$ which corresponds to $\alpha_n$ of order $\tfrac{n \pi}{L}$.

Consider first $2 \le n \le L - 2$. If
\[
 \alpha_n^+ = \tfrac{n \pi}{L - 1} ,
\]
we have
\[
 \cos(\tfrac{L - 2}{2} \alpha_n^+ - \tfrac{n \pi}{2}) = \cos(\tfrac{n \pi}{2 L - 2}) = \cos(\tfrac{L}{2} \alpha_n^+ - \tfrac{n \pi}{2}) .
\]
On the other hand, if
\[
 \alpha_n^- = \tfrac{(n - 1) \pi}{L - 2} ,
\]
then
\[
 \cos(\tfrac{L - 2}{2} \alpha_n^- - \tfrac{n \pi}{2}) = 0 < \cos(\tfrac{(n - 1) \pi}{L - 2} - \tfrac{\pi}{2}) = \cos(\tfrac{L}{2} \alpha_n^- - \tfrac{n \pi}{2}) .
\]
Hence,
\[
 \frac{\cos(\tfrac{L - 2}{2} \alpha_n^+ - \tfrac{n \pi}{2})}{\cos(\tfrac{L}{2} \alpha_n^+ - \tfrac{n \pi}{2})} = 1 \ge \delta > 0 = \frac{\cos(\tfrac{L - 2}{2} \alpha_n^- - \tfrac{n \pi}{2})}{\cos(\tfrac{L}{2} \alpha_n^- - \tfrac{n \pi}{2})} .
\]
It follows that there is a solution $f_n$ corresponding to $\alpha_n \in (\alpha_n^-, \alpha_n^+]$.

Finally, when $n = 1$ and $\delta > 1 - \tfrac{2}{L}$, we have $\alpha_n^- = 0$, and
\[
 \frac{\cos(\tfrac{L - 2}{2} \alpha_n^+ - \tfrac{n \pi}{2})}{\cos(\tfrac{L}{2} \alpha_n^+ - \tfrac{n \pi}{2})} = 1 \ge \delta > 1 - \tfrac{2}{L} = \frac{\cos(\tfrac{L - 2}{2} \alpha_n^- - \tfrac{n \pi}{2})}{\cos(\tfrac{L}{2} \alpha_n^- - \tfrac{n \pi}{2})}
\]
(with the right-hand side understood in the sense of continuous extension), so once again we have a solution $f_n$ corresponding to $\alpha_n \in (\alpha_n^-, \alpha_n^+]$. (Recall that if $\delta \le 1 - \tfrac{2}{L}$, then there is a corresponding $\sinh$ solution $f_1$.) This ends the proof.
\end{proof}
We shall now prove $\ell^\infty$ estimates for the normalized eigenvectors. We have the following proposition.

\begin{proposition}
\label{prop:3:8:a}
Let $\tilde f_n = f_n / \lVert f_n \rVert$ be the normalized eigenvector of $h_\delta$. If $2 \le n \le L - 2$, we have
\[
 \lvert \tilde f_n(X) \rvert \le \frac{2}{\sqrt{L - 1}} .
\]
Furthermore, if $K$ is fixed, then there exists a constant $C$ (which depends only on $K$) such that for $L$ large enough and every $n$ we have
\[
 \lvert \tilde f_n(X) \rvert \le \frac{C}{\sqrt{L - 1}} .
\]
\end{proposition}
\begin{proof}
For the regular solutions, we clearly have
\[
 \lvert f_n(X) \rvert \le 1 .
\]
On the other hand,
\[\begin{split}
 \lVert f_n \rVert_2^2 & = \sum_{X = 1}^{L - 1} \cos^2(\tfrac{(n - 2) \pi}{2} + \alpha_n (X - \tfrac{L}{2})) \\
 & = \sum_{X = 1}^{L - 1} \bigl(\tfrac{1}{2} + \tfrac{1}{2} (-1)^n \cos(\alpha_n (2 X - L))\bigr) \\
 & = \frac{L - 1}{2} + (-1)^n \frac{\sin((L - 1) \alpha_n)}{2 \sin(\alpha_n)} .
\end{split}\]
If $n \ge 2$, we have $\tfrac{\pi}{L - 1} \le \alpha_n \le \pi - \tfrac{\pi}{L - 1}$, and hence $\sin \alpha_n \ge \tfrac{2}{L - 1}$ and
\[
 \lVert f_n \rVert_2^2 \ge \frac{L - 1}{2} - \frac{1}{2 \sin(\alpha_n)} \ge \frac{L - 1}{2} - \frac{L - 1}{4} \ge \frac{L - 1}{4}
\]
when $L \ge 4$. This gives the desired bound.

A different argument is needed when $n = 0$ or $n = 1$. For $n = 0$, we have
\begin{align}
\notag
 \lVert f_0 \rVert_2^2 & = \sum_{X = 1}^{L - 1} \cosh^2(\beta_0 (X - \tfrac{L}{2})) \\
\label{eq:beta0:lower}
 & = \sum_{X = 1}^{L - 1} \bigl(\tfrac{1}{2} + \tfrac{1}{2} \cosh(\beta_0 (2 X - L))\bigr) \ge L - 1 .
\end{align}
On the other hand, if $L \ge 3 \lvert K \rvert$, we have $\lvert \frac{\pi K}{L} \rvert \le \tfrac{\pi}{3}$. There is a constant $C$ such that
\[ \log \frac{1}{\cos x} \le C x^2 \]
when $-\tfrac{\pi}{3} \le x \le \tfrac{\pi}{3}$, and so
\[ \log \frac{1}{\delta} = \log \frac{1}{\cos \frac{\pi K}{L}} \le C \frac{K^2}{L^2}. \]
Recall that $\beta_0 \le \operatorname{arcosh}(1 / \delta)$ and $\delta = \cos \tfrac{\pi K}{L} \ge \frac{1}{2}$. For some constant $C$,
\[ \operatorname{arcosh} e^t \le C \sqrt{t} \]
when $0 \le t \le \log 2$. Thus, since $\delta \ge \tfrac{1}{2}$, we get
\[ \beta_0 \le C \sqrt{\log \frac{1}{\delta}} \]
and so 
\[ \beta_0 L \leq C \lvert K \rvert \]
with an absolute constant. This gives
\[
 \lvert f_0(X) \rvert \le \cosh(\beta_0 \tfrac{L - 2}{2}) \le \cosh \tfrac{C \lvert K \rvert}{2} .
\]
Together with~\eqref{eq:beta0:lower}, we get the desired bound on $\tilde f_0$.

We now turn to $n = 1$. If $L$ is large enough, $1 - \delta = 1 - \cos \tfrac{\pi K}{L} \le \tfrac{1}{2} (\tfrac{\pi K}{L})^2 < \tfrac{2}{L}$, and hence $f_1$ is the regular solution. Since $\alpha_1 \le \frac{\pi}{L - 1}$, we have $\alpha_1 \lvert X - \tfrac{L}{2} \rvert < \frac{\pi}{2}$. It follows that
\[
 \lvert f_1(X) \rvert \ge \tfrac{2}{\pi} \alpha_1 \lvert X - \tfrac{L}{2} \rvert ,
\]
and hence
\[
 \lVert f_1 \rVert_2^2 \ge \frac{4 \alpha_1^2}{\pi^2} \sum_{X = 1}^{L - 1} \lvert X - \tfrac{L}{2} \rvert^2 = \frac{L (L - 1) (L - 2) \alpha_1^2}{3 \pi^2} .
\]
On the other hand,
\[
 \lvert f_1(X) \rvert \le \alpha_1 \tfrac{L - 1}{2} ,
\]
and again the desired estimate on $\tilde f_1$ follows.
\end{proof}

Observe that when $n \ge 2$, we have
\[
 \lvert f_n(1) \rvert = \lvert \cos(\alpha_n \tfrac{L - 2}{2} - \tfrac{n \pi}{2}) \rvert .
\]
Using $\alpha_n^- \le \alpha_n \le \alpha_n^+$, we find that
\[
 \tfrac{(n - 1) \pi}{2} - \tfrac{n \pi}{2} \le \alpha_n \tfrac{L - 2}{2} - \tfrac{n \pi}{2} \le \tfrac{n (L - 2) \pi}{2 (L - 1)} - \tfrac{n \pi}{2} ,
\]
which simplifies to
\[
 -\tfrac{\pi}{2} \le \alpha_n \tfrac{L - 2}{2} - \tfrac{n \pi}{2} \le -\tfrac{n \pi}{2 (L - 1)} .
\]
Therefore,
\[
 \lvert f_n(1) \rvert \le \cos(\tfrac{n \pi}{2 (L - 1)}) \le \tfrac{\pi}{2} - \tfrac{n \pi}{2 (L - 1)} = \tfrac{(L - 1 - n) \pi}{2 (L - 1)} .
\]
We have thus proved the following result.

\begin{proposition}
\label{prop:3:8:b}
Let $\tilde f_n = f_n / \lVert f_n \rVert$ be the normalized eigenvector of $h_\delta$. If $2 \le n \le L - 2$, we have
\[
 \lvert \tilde f_n(1) \rvert \le \tfrac{(L - 1 - n) \pi}{(L - 1)^{3 / 2}} .
\]
\end{proposition}

We shall now compare the eigenvalues of $h_\delta$ with the energies \eqref{E_K,p} coming from the spin wave approximation. 

\begin{proposition} \label{prop:hdeltaeigenvaluescomparison}
Define
$$ E_n = 2 - 2 \delta \cos(\frac{n\pi}{L-1}) , \qquad E_n' = \tfrac{4 \pi}{L - 1} \delta \sin(\frac{n\pi}{L-1}) .$$
Then,  if $2 \le n \le L - 2$, we have
\begin{equation*}
 \biggl\lvert \frac{\lambda_n - E_n}{E_n'} \biggr\rvert \le 1 .
\end{equation*}
\end{proposition}
\begin{proof}
    Observe that for $n \ge 2$,
\[
 \alpha_n^+ - \alpha_n^- = \tfrac{n \pi}{L - 1} - \tfrac{(n - 1) \pi}{L - 2} = \tfrac{(L - 1 - n) \pi}{(L - 1) (L - 2)} \le \tfrac{\pi}{L - 1} .
\]
Thus
\[
 \lambda_n - E_n = 2 \delta (\cos(\alpha_n^+) - \cos(\alpha_n)) = -4 \delta \sin(\tfrac{1}{2} \alpha_n + \tfrac{1}{2} \alpha_n^+) \sin(\tfrac{1}{2} \alpha_n^+ - \tfrac{1}{2} \alpha_n) ,
\]
and so
\[
 \lvert \lambda_n - E_n \rvert \le 2 \delta \sin(\tfrac{1}{2} \alpha_n + \tfrac{1}{2} \alpha_n^+) (\alpha_n^+ - \alpha_n^-) .
\]
It follows that
\[
 \biggl\lvert \frac{\lambda_n - E_n}{E_n'} \biggr\rvert \le \frac{\sin(\tfrac{1}{2} \alpha_n + \tfrac{1}{2} \alpha_n^+)}{2 \sin(\alpha_n^+)} .
\]
If $n \le \tfrac{L - 1}{2}$, then $\tfrac{1}{2} \alpha_n + \tfrac{1}{2} \alpha_n^+ \le \alpha_n^+ \le \tfrac{\pi}{2}$, and so the right-hand side of the above inequality does not exceed $\tfrac{1}{2}$. If $n > \tfrac{L - 1}{2}$, then 
\[
 \biggl\lvert \frac{\lambda_n - E_n}{E_n'} \biggr\rvert \le \frac{\sin(\alpha_n^+) + (\alpha_n^+ - \alpha_n^-)}{2 \sin(\alpha_n^+)} .
\]
However,
\[
 \sin(\alpha_n^+) = \sin(\tfrac{n \pi}{L - 1}) \ge 2 - \tfrac{2 n}{L - 1} = \tfrac{2 (L - 1 - n)}{L - 1} \ge \alpha_n^+ - \alpha_n^- ,
\]
and consequently
\[
 \biggl\lvert \frac{\lambda_n - E_n}{E_n'} \biggr\rvert \le 1 .
\]
This ends the proof.
\end{proof}

We conclude this section by linking the spectral analysis of $h_\delta$ in \eqref{eq:h_definition} with that of \eqref{eq_onebody_Ham}. The latter operator is defined on the Hilbert space $\mathcal{H}_{K,L}^{(1)}$ defined in \eqref{H_1_def}. This space has additional symmetry constraints. Taking them into account yields the following result. 
\begin{corollary}\label{cor:lambdakj}
   For $j\in \mathcal{J}_{K,L}:= \{1,2,\ldots, \dim \mathcal{H}_{K,L}^{(1)}\},$ let $\lambda_{K,j}$ and $f_{K,j}$ denote the eigenvalues and, respectively, eigenvectors of $h_{\Lambda,K}$ defined in \eqref{eq_onebody_Ham}. Then 
   \begin{equation}
\lambda_{K, j} = \begin{cases}
\lambda_{2 j - 2} & \text{if } K \text{ is even,} \\[4pt]
\lambda_{2 j - 1} & \text{if } K \text{ is odd,}
\end{cases}
\qquad \text{and} \qquad
f_{K, j} = \begin{cases}
\dfrac{f_{2 j - 2}}{\|f_{2 j - 2}\|} & \text{if } K \text{ is even,} \\[8pt]
\dfrac{f_{2 j - 1}}{\|f_{2 j - 1}\|} & \text{if } K \text{ is odd.}
\end{cases}
\end{equation}
\end{corollary}
In Section \ref{sec:main-kroof} we will use the following version of the results obtained above.

\begin{proposition}\label{prop:eigenvalue_comparison_spin}
Fix \(K\in\mathbb Z\).  For \(j\in\mathcal J_{K,L}\), define
$$
    E_{K,j}
    :=
    2-2\delta\cos\left(\frac{2\pi j}{L}\right)
$$
and
$$
    E'_{K,j}
    :=
    \frac{4\pi\delta}{L}
    \sin\left(\frac{2\pi j}{L}\right).
$$
Then, when $L$ is large enough (depending on $K$), for every $j\in\mathcal J_{K,L}\setminus\{L/2\}$, we have
\begin{equation}
\label{eq:eigenvalue_comparison_spin}
    \left|
        \frac{\lambda_{K,j}-E_{K,j}}{E'_{K,j}}
    \right|
    \leq C_K .
\end{equation}
\end{proposition}

\begin{proof}
For \(j\in\mathcal J_{K,L}\), define
$$
    n_K(j):=
    \begin{cases}
        2j-2, & K \text{ even},\\
        2j-1, & K \text{ odd}.
    \end{cases}
$$
By Corollary~\ref{cor:lambdakj},
$$
    \lambda_{K,j}=\lambda_{n_K(j)}.$$
For brevity, write
$$
    n:=n_K(j),\qquad
    a:=\frac{n\pi}{L-1},
    \qquad
    q:=\frac{2\pi j}{L}.$$
Obviously
\begin{equation}
\label{eq:two_grids_distance}
    |a-q|\leq \frac{2\pi}{L}.
\end{equation}
Recall the quantities introduced in Proposition~\ref{prop:hdeltaeigenvaluescomparison}:
$$
    E_n=2-2\delta\cos(a),
    \qquad
    E'_n=\frac{4\pi\delta}{L-1}\sin(a).
$$
We first consider the case \(n\geq2\). By
Proposition~\ref{prop:hdeltaeigenvaluescomparison},
\begin{equation}
\label{eq:comparison_with_auxiliary_energy}
    |\lambda_n-E_n|\leq E'_n.
\end{equation}

Since
\(j\in\mathcal J_{K,L}\setminus\{L/2\}\), we have \(0<q<\pi\),
and the distance of \(q\) from the endpoints \(0\) and \(\pi\)
is at least \(\pi/L\). Hence
\begin{equation}
\label{eq:sin_q_lower_bound}
    \sin(q)\geq \frac{2}{L}.
\end{equation}
Using \eqref{eq:two_grids_distance}, we therefore obtain
$$
    \sin(a)
    \leq \sin(q)+|a-q|
    \leq (1+\pi)\sin(q).
$$
It follows that
\begin{equation}
\label{eq:derivative_energy_comparison}
    E'_n
    =
    \frac{4\pi\delta}{L-1}\sin(a)
    \leq C\frac{4\pi\delta}{L}\sin(q)
    =
    C E'_{K,j}.
\end{equation}
Furthermore,
\begin{align}
    |E_n-E_{K,j}|
    &=
    2\delta|\cos(a)-\cos(q)|
    \nonumber\\
    &\leq
    2\delta |a-q|
    \left(
        \sin(q)+\frac12|a-q|
    \right)
    \nonumber\leq
    C\delta
    \left(
        \frac{\sin(q)}{L}+\frac1{L^2}
    \right).
\label{eq:energy_grid_comparison}
\end{align}
By \eqref{eq:sin_q_lower_bound}, the second term in the last
parenthesis can be absorbed into the first. Thus
\begin{equation}
\label{eq:energy_grid_comparison_final}
    |E_n-E_{K,j}|
    \leq
    C\frac{\delta}{L}\sin(q)
    \leq C E'_{K,j}.
\end{equation}
Combining
\eqref{eq:comparison_with_auxiliary_energy},
\eqref{eq:derivative_energy_comparison}, and
\eqref{eq:energy_grid_comparison_final}, we conclude that
\[
 |\lambda_{K,j}-E_{K,j}|
    \leq
    |\lambda_n-E_n|+|E_n-E_{K,j}|   
    \leq C E'_{K,j}
\]
whenever \(n_K(j)\geq2\).

It remains to treat the lowest mode \(j=1\). If \(K\) is even,
then \(n_K(1)=0\), while if \(K\) is odd, then \(n_K(1)=1\).
Let \(h_1\) denote the operator \(h_{\delta}\) in \eqref{eq:h_definition} with
\(\delta=1\). This is the discrete Laplacian on $\{1,\ldots,L-1\}$ with Neumann boundary conditions. Its eigenvalues can be obtained by a direct cosine ansatz. They are
$$    \lambda_n(1)
    =
    2-2\cos\left(\frac{n\pi}{L-1}\right),
    \qquad n=0,\ldots,L-2.$$
Moreover,
$$
    \|h_{\delta}-h_1\|
    \leq 2(1-\delta).
$$
Hence, by Weyl's inequality,
$$
    |\lambda_n-\lambda_n(1)|
    \leq 2(1-\delta).
$$
Since $K$ is fixed,
$$ 1-\delta
    =
    1-\cos\left(\frac{\pi K}{L}\right)
    \leq \frac{\pi^2K^2}{2L^2}.
$$
For $n\in\{0,1\}$, both
$$
    \frac{n\pi}{L-1}
    \qquad\text{and}\qquad
    \frac{2\pi}{L}
$$
are of order \(L^{-1}\). Therefore,
$$
    |\lambda_n(1)-E_{K,1}|
    \leq \frac{C_K}{L^2},
$$
and consequently
$$
    |\lambda_{K,1}-E_{K,1}|
    \leq \frac{C_K}{L^2}.
$$
On the other hand, for $L$ sufficiently large,
$\delta\geq1/2$, and hence
$$
    E'_{K,1}
    =
    \frac{4\pi\delta}{L}
    \sin\left(\frac{2\pi}{L}\right)
    \geq \frac{c}{L^2}.
$$
It follows that
\[
    |\lambda_{K,1}-E_{K,1}|
    \leq C_K E'_{K,1}.
\]
This proves the claim for every
\(j\in\mathcal J_{K,L}\setminus\{L/2\}\).
\end{proof}

\section{Proof of Main Result}\label{sec:main-kroof}
\label{chapter:proof}
In this section we prove Theorem~\ref{thm:wave-kacket}. We will use the effective one-body Hamiltonian $h_{\Lambda,K}$, introduced in  Section \ref{section:one-body-reduction} to rewrite the time evolution of our initial state. We will then use the spectral properties of $h_{\Lambda,K}$, investigated in Section~\ref{section:spectral problem}, to obtain the appropriate bounds on the norm difference. 
\begin{proof}[Proof of Theorem~\ref{thm:wave-kacket}]
Fix $K \in \Z$ and assume $\mathbb{N}\ni L \gg 1$. Let $\Lambda =\{1,2,\ldots,L\}$ and let $\mathcal{K}_K$ be the set of relative momenta at fixed total momentum $K$ defined by \eqref{def:K_K}. For $k\in\mathcal{K}_K$ define 
\begin{displaymath}
    w_k(X) = \cos\left(\frac{2\pi k}{L}X\right) \qquad X\in\{1,2,\ldots,L-1\},
\end{displaymath}
and, for $t\geq 0$, put 
\begin{displaymath}
    \begin{split}
    &f^t_K(X) = (e^{-i h_{\Lambda,K} t}f_K^0)(X) \quad \text{for} \quad 
    f_K^0(X) = \frac{1}{L} \sum_{k\in\mathcal{K}_K} c_{K,k} w_k(X),\\
    &f^{t,\mathrm{free}}_K(X) = \frac{1}{L} \sum_{k\in\mathcal{K}_K} c_{K,k} e^{-i E_{K,k} t} w_k(X),
    \end{split}
\end{displaymath}
where (see \eqref{E_K,p})
\begin{displaymath}
    E_{K,k}= 2\left(1-\cos\left(\tfrac{\pi K}{L}\right)\cos\left(\tfrac{2\pi k}{L}\right)\right)
\end{displaymath}
is our approximation for the energy of a two-spin wave state (see the discussion in Section \ref{section:spin waves}),
and the one-body Hamiltonian $h_{\Lambda,K}$ is given by \eqref{eq_onebody_Ham}.  
Since $w_k\in \mathcal{H}_{K,L}^{(1)}$, we also have $f^0_K, f^t_K, f^{t,\mathrm{free}}_K \in \mathcal{H}_{K,L}^{(1)}$. These functions may be extended to all of $\Z$ using the symmetries \eqref{eq:symmetries_of_f}.
According to the discussion in Section \ref{section:time evolution for wave function}, the time evolutions $\Psi^t_K$ and $\Psi_K^{t,\mathrm{free}}$, defined in \eqref{psi_K^t-intro}  can now be written as
\begin{displaymath}
    \Psi_K^t = \sum_{\substack{x,y\in\Lambda \\ x\neq y}} \psi_K^t(x,y) S^+_x S^+_y \vac, \qquad 
    \Psi_K^{t,\mathrm{free}} = \sum_{\substack{x,y\in\Lambda \\ x\neq y}} \psi_K^{t,\mathrm{free}}(x,y) S^+_x S^+_y \vac,
\end{displaymath}
where, from \eqref{eq:Schrodinger_psi_small}, \eqref{special_case_psi_K^0} and \eqref{e^h_lambda,P}, we have
\begin{displaymath}
\begin{split}
&\psi_K^t(x,y) = e^{-i h_{\Lambda}t} \psi_K^0(x,y) = e^{\frac{2\pi i}{L}K\left(\frac{x+y}{2}\right)}f_K^t(x-y),\\
&\psi_K^{t,\mathrm{free}}(x,y) = e^{\frac{2\pi i}{L}K\left(\frac{x+y}{2}\right)}f_K^{t,\mathrm{free}}(x-y).
\end{split}
\end{displaymath}
We start by discussing the normalization of the initial state. Using \eqref{eq:spinscalarproduct} and the notation introduced above, we find that
\begin{equation}
\label{eq:normalization}
\begin{split}
\| \Psi_K^0 \|_{\mathcal{H}_\Lambda}^2 &= 2 \sum_{\substack{x,y\in\Lambda \\ x\neq y}} |\psi_K^0(x,y)|^2 = 2 \sum_{\substack{x,y\in\Lambda \\ x\neq y}}  |f_K^0(x-y)|^2 = 2 L \sum_{X=1}^{L-1} |f_K^0(X)|^2\\
&= \frac{2}{L} \sum_{k\in \mathcal{K}_K } \sum_{{k'} \in \mathcal{K}_K } \overline{c_{K,k}} c_{K,k'} \sum_{X=1}^{L-1} \cos\left(\tfrac{2\pi k}{L} X\right)\cos\left(\tfrac{2\pi{k'}}{L}X\right) \\ 
&= \frac{2}{L} \sum_{k\in \mathcal{K}_K } \sum_{{k'} \in \mathcal{K}_K } \overline{c_{K,k}} c_{K,k'}
\left(\frac{L}{2}\left(\delta_{k+{k'},0} + \delta_{k,{k'}}\right) - 1 \right) \\
&= {\1_{K\;\text{even}}}|c_{K,0}|^2  + \sum_{k\in \mathcal{K}_K } |c_{K,k}|^2 - \frac{2}{L} \left| \sum_{k\in \mathcal{K}_K } c_{K,k} \right|^2. 
\end{split}
\end{equation}
 Recall that, by assumption, the coefficients $c_{K,k}$ are chosen in such a way that the above expression is equal to one. 
 
 Let us now evaluate the norm difference $\|\Psi^t_K - \Psi_K^{t,\mathrm{free}}\|$. Since $h_{\Lambda,K}$ is self-adjoint on $\mathcal{H}_{K,L}^{(1)}$, its normalized eigenvectors $\{f_{K,j} \left. | \right. j \in \mathcal{J}_{K,L} \}$ form an orthonormal basis of $\mathcal{H}_{K,L}^{(1)}$.   Using the notation introduced above we obtain
\begin{equation}
\label{norm_diff_step_1}
\begin{split}
\left\| \Psi_K^t  - \Psi_K^{t,\mathrm{free}} \right\|_{\mathcal{H}_\Lambda}^2 
&= 2 \langle \psi_K^t - \psi_K^{t,\mathrm{free}}\left. | \right. \psi_K^t - \psi_K^{t,\mathrm{free}} \rangle_{\mathfrak{h}_\Lambda} \\
&= 2 \sum_{\substack{x,y \in \Lambda \\ x \neq y}} \left| f_K^t(x-y) - f_K^{t,\mathrm{free}}(x-y) \right|^2 \\
&= 2 L \sum_{X=1}^{L-1} \left| f_K^t(X) - f_K^{t,\mathrm{free}}(X) \right|^2\\
&= 2 L  \langle f_K^t - f_K^{t,\mathrm{free}} \left. | \right. f_K^t - f_K^{t,\mathrm{free}} \rangle_{\mathcal{H}_{K,L}^{(1)}}\\ 
&= 2 L \sum_{j\in\mathcal{J}_{K,L}} \left| \frac{1}{L} \sum_{k\in \mathcal{K}_K} c_{K,k} \left\langle f_{K,j} \left. \Big| \right.  e^{-i h_{\Lambda,K}t} w_k - e^{-i E_{K,k}t} w_k \right\rangle_{\mathcal{H}_{K,L}^{(1)}}\right|^2 \\
&= \frac{2}{L}  \sum_{j\in \mathcal{J}_{K,L}} \left| \sum_{k\in \mathcal{K}_K} c_{K,k} \left( 1 - e^{i(\lambda_{K,j} - E_{K,k})t}\right) \langle f_{K,j} \left. | \right. w_k \rangle_{\mathcal{H}_{K,L}^{(1)}} \right|^2.
\end{split}
\end{equation}
where we used the identity decomposition $\1_{\mathcal{H}_{K,L}^{(1)}} = \sum_{j\in \mathcal{J}_{K,L}} |f_{K,j}\rangle \langle f_{K,j}|$. Note that only the terms where $\lambda_{K,j} \neq E_{K,k}$ contribute to the above sum. The overlap  $\langle f_{K,j} \left. | \right. w_k \rangle_{\mathcal{H}_{K,L}^{(1)}}$ can be found by using the expression \eqref{eq_onebody_Ham}. Indeed, we consider three cases:
\begin{enumerate}
\item $X \in \{2,\ldots,L-2\}$. Then:
    \begin{displaymath}
    \begin{split}
    (h_{\Lambda,K}w_k)(X) &= 2 w_k(X) - \cos\left(\tfrac{\pi K}{L} \right)(w_k(X+1) + w_k(X-1)) \\
    &= 2\cos\left(\tfrac{2\pi k}{L}X \right) - \cos\left(\tfrac{\pi K}{L} \right)\left(\cos\left(\tfrac{2\pi k}{L}(X+1)\right) + \cos\left(\tfrac{2\pi k}{L}(X- 1)\right) \right) \\
    &= 2\cos\left(\tfrac{2\pi k}{L}X \right)-2\cos\left(\tfrac{2\pi k}{L}X \right)\cos\left(\tfrac{\pi K}{L}\right)\cos\left(\tfrac{2\pi k}{L} \right)\\
    &= 2\left(1-\cos\left(\tfrac{\pi K}{L}\right)\cos\left(\tfrac{2\pi k}{L} \right) \right)\cos\left(\tfrac{2\pi k}{L}X \right) = E_{K,k}w_k(X).
    \end{split}
\end{displaymath}
\item $X = 1$. Similar calculations show that $(h_{\Lambda,K}w_k)(1) = E_{K,k}w_k(1) + \Gamma_{K,k}$  where $\Gamma_{K,k}$ has been introduced in \eqref{def:GammaKkE'Kk}.
\item $X = L-1$. According to \eqref{def:K_K},  $\cos(2\pi k) = (-1)^K$ for $k\in \mathcal{K}_K$. Therefore 
\begin{displaymath}
    (h_{\Lambda,K}w_k)(L-1)= E_{K,k}w_k(L-1) + \cos(2\pi k) \Gamma_{K,k} = E_{K,k}w_k(L-1) + (-1)^K \Gamma_{K,k}.
\end{displaymath}
\end{enumerate}
Putting these results together, we write
\begin{displaymath}
    \left((h_{\Lambda,K} - E_{K,k})w_k\right)(X) = \Gamma_{K,k} \left( \delta_{X,1} + (-1)^K \delta_{X,L-1} \right)\qquad  \forall\,X\in\{1,2,\ldots,L-1\}.
\end{displaymath}
Note that \eqref{H_1_def} implies $f_{K,j}(L-1) = (-1)^K f_{K,j}(1)$, hence
\begin{align}\label{exact_f_and_cos_overlap}
    \begin{split}
    \langle f_{K,j} \left. | \right. w_k  \rangle_{\mathcal{H}_{K,L}^{(1)}} 
&= \frac{\langle (h_{\Lambda,K} - E_{K,k}) f_{K,j} \left. | \right. w_k \rangle_{\mathcal{H}^{(1)}_{K,L}}}{\lambda_{K,j} - E_{K,k}}  \\ 
&= \frac{\langle f_{K,j} \left. | \right. (h_{\Lambda,K} - E_{K,k})   w_k \rangle_{\mathcal{H}^{(1)}_{K,L}}}{\lambda_{K,j} - E_{K,k}}   =  \frac{ 2 \overline{f_{K,j}(1)}\Gamma_{K,k}}{\lambda_{K,j} - E_{K,k}}\,,
    \end{split}
\end{align}
so that \eqref{norm_diff_step_1} becomes
\begin{equation}
\label{norm_difsec:main-krooff_2}
\left\| \Psi_K^t - \Psi_K^{t,\mathrm{free}} \right\|_{\mathcal{H}_\Lambda}^2 = \frac{8}{L} \sum_{j \in \mathcal{J}_{K,L}} |f_{K,j}(1)|^2 \left| \sum_{k\in \mathcal{K}_K} c_{K,k}\frac{1 - e^{i(\lambda_{K,j} - E_{K,k})t}}{\lambda_{K,j} - E_{K,k}}\,\Gamma_{K,k}\right|^2
\end{equation}
(terms with $\lambda_{K,j} - E_{K,k} = 0$ are understood in the limiting sense as $c_{K,k} (-it) \Gamma_{K,k}$).

To get rid of the double sum in \eqref{norm_difsec:main-krooff_2} without losing the information about the timescales, we use suitable substitutions for the expressions inside the absolute values, which will be later rewritten in terms of the norm of the action of appropriately defined uniformly bounded operators on some vectors (cf. \eqref{T:action}). 

For $j\in\mathcal J_{K,L}$, we split the expression inside the absolute value in \eqref{norm_difsec:main-krooff_2} as
$$
A_{K,j}^{(0)}(t)+A_{K,j}^{(\neq 0)}(t),
$$
where

$$
A_{K,j}^{(0)}(t)
:=
\1_{|K|\ \mathrm{even}}\,
c_{K,0}\Gamma_{K,0}
\frac{1-e^{i(\lambda_{K,j}-E_{K,0})t}}
{\lambda_{K,j}-E_{K,0}}
$$
and
$$
A_{K,j}^{(\neq 0)}(t)
:=
\sum_{k\in\mathcal K_K\setminus\{0\}}
c_{K,k}\Gamma_{K,k}
\frac{1-e^{ i(\lambda_{K,j}-E_{K,k})t}}
{\lambda_{K,j}-E_{K,k}}.
$$
We shall first deal with $A_{K,j}^{(0)}(t)$ which corresponds to the $k=0$ term in the sum in \eqref{norm_difsec:main-krooff_2} and has to be estimated separately for reasons that will become clear later. The term $\1_{|K|\ \mathrm{even}}$ follows from \eqref{def:K_K}. Since
\begin{equation}\label{eq:integralfrac}
\frac{1-e^{i(\lambda_{K,j}-E_{K,0})t}}
{\lambda_{K,j}-E_{K,0}}=- i
\int_0^t
e^{ i(\lambda_{K,j}-E_{K,0})s}\,ds,
\end{equation}
we have 
$$
\bigl|A_{K,j}^{(0)}(t)\bigr|
\leq
t\,|c_{K,0}\Gamma_{K,0}|.
$$
Consequently
\begin{equation}\label{eq:R0bound}
\begin{aligned}
R_0(t)
&:=
\1_{|K|\ \mathrm{even}}\,
\frac{8}{L}
\sum_{j\in\mathcal J_{K,L}}
|f_{K,j}(1)|^2
\bigl|A_{K,j}^{(0)}(t)\bigr|^2
\\
&\leq
\1_{|K|\ \mathrm{even}}\,
\frac{8t^2}{L}
|c_{K,0}\Gamma_{K,0}|^2
\sum_{j\in\mathcal J_{K,L}}
|f_{K,j}(1)|^2
\leq
\1_{|K|\ \mathrm{even}}\,
\frac{8t^2}{L}
|c_{K,0}\Gamma_{K,0}|^2.
\end{aligned}
\end{equation}
Indeed, since
\(\{f_{K,j}\}_{j\in\mathcal J_{K,L}}\) is an orthonormal basis of
\(\mathcal H_{K,L}^{(1)}\), Parseval's identity gives
$$
\sum_{j\in\mathcal J_{K,L}}
|f_{K,j}(1)|^2
=
\sum_{j\in\mathcal J_{K,L}}
\bigl|\langle f_{K,j},\delta_1\rangle\bigr|^2
=
\|P_{K,L}\delta_1\|^2
\leq 1,
$$
where \(P_{K,L}\) denotes the orthogonal projection onto
$\mathcal H_{K,L}^{(1)}$. We next observe that the normalization condition \eqref{eq:normalization} implies
$$
|c_{K,0}|\leq 1.
$$
Indeed, let
$$
N_K:=|\mathcal K_K|,
\qquad
S_K:=\sum_{k\in\mathcal K_K}|c_{K,k}|^2.
$$
If $|K|$ is even, then, by the Cauchy--Schwarz inequality,
$$
\begin{aligned}
1
&=
S_K+|c_{K,0}|^2
-\frac{2}{L}
\left|
\sum_{k\in\mathcal K_K}c_{K,k}
\right|^2
\geq
S_K+|c_{K,0}|^2-\frac{2N_K}{L}S_K
\\
&=
|c_{K,0}|^2
+
\left(1-\frac{2N_K}{L}\right)S_K \geq |c_{K,0}|^2,
\end{aligned}
$$
where in the last step we used \(N_K\leq L/2\). Finally,
$$
|\Gamma_{K,0}|
=
\left|
\cos\left(\frac{\pi K}{L}\right)-1
\right|
=
1-\cos\left(\frac{\pi K}{L}\right)
\leq
\frac{\pi^2K^2}{2L^2}.
$$
Combining this estimate with \eqref{eq:R0bound}, we conclude that, for fixed $K$,
\begin{equation} \label{eq:R0boundfinal}
R_0(t)
\leq
\frac{2\pi^4K^4t^2}{L^5}.
\end{equation}
We shall now deal with the terms related to $
A_{K,j}^{(\neq 0)}(t)$. First we isolate the term corresponding to $j=L/2$. By \eqref{def:K_K}, this value belongs to $\mathcal J_{K,L}$ only when $L$ and $|K|$ are both even. After the $k=0$ contribution has been separated, define
\begin{equation} \label{eq:RL/2}
R_{L/2}(t)
:=
\1_{L\ \mathrm{even}}
\1_{|K|\ \mathrm{even}}
\frac{8}{L}
|f_{K,L/2}(1)|^2
\left|
\sum_{k\in\mathcal K_K\setminus\{0\}}
c_{K,k}\Gamma_{K,k}
\frac{
1-e^{i(\lambda_{K,L/2}-E_{K,k})t}
}{
\lambda_{K,L/2}-E_{K,k}
}
\right|^2.
\end{equation}
Since $L$ and $|K|$ are even, Corollary \ref{cor:lambdakj} and Proposition \ref{prop:3:8:b} imply
$$
|f_{K,L/2}(1)|^2
\leq
\frac{C}{L^3}.
$$
Using the integral representation as in \eqref{eq:integralfrac} and Cauchy--Schwarz inequality, we obtain 
\begin{equation} \label{eq:RL/2finalbound}
\begin{aligned}
R_{L/2}(t)
&\leq
\frac{C}{L^4}
\left(
\sum_{k\in\mathcal K_K\setminus\{0\}}
|c_{K,k}\Gamma_{K,k}|^2
\right)
\left(
\sum_{k\in\mathcal K_K\setminus\{0\}}
\left|
\int_0^t
e^{i(\lambda_{K,L/2}-E_{K,k})s}\,ds
\right|^2
\right)
\\
&\leq
\frac{Ct^2}{L^4}
\left(
\sum_{k\in\mathcal K_K\setminus\{0\}}
|c_{K,k}\Gamma_{K,k}|^2
\right)
|\mathcal K_K\setminus\{0\}|
\\
&\leq
\frac{Ct^2}{L^3}
\sum_{k\in\mathcal K_K\setminus\{0\}}
|c_{K,k}\Gamma_{K,k}|^2 \leq
C\frac{\varepsilon^4t^2}{L^3}.
\end{aligned}
\end{equation}
where in the last step we used \eqref{c_assumption_1}. It remains to consider the contribution in which $k\neq 0$ and $j\neq L/2$. We denote this contribution by
$$
\begin{aligned}
R_{\neq 0}(t)
:=
\frac{8}{L}
\sum_{j\in\mathcal J_{K,L}\setminus\{L/2\}}
|f_{K,j}(1)|^2
\left|
\sum_{k\in\mathcal K_K\setminus\{0\}}
c_{K,k}\Gamma_{K,k}
\frac{
1-e^{i(\lambda_{K,j}-E_{K,k})t}
}{
\lambda_{K,j}-E_{K,k}
}
\right|^2.
\end{aligned}
$$
When $|K|$ is even, it may happen that $j=k$. We isolate this term and, as before, use
$$
\frac{
1-e^{ i(\lambda_{K,j}-E_{K,j})t}
}{
\lambda_{K,j}-E_{K,j}
}
=
-i
\int_0^t
e^{i(\lambda_{K,j}-E_{K,j})s}\,ds.
$$
For the remaining terms, for which $k\neq j$, we write
$$
\frac{
1-e^{i(\lambda_{K,j}-E_{K,k})t}
}{
\lambda_{K,j}-E_{K,k}
}
=
\frac{
1-e^{ i(\lambda_{K,j}-E_{K,k})t}
}{
E_{K,j}-E_{K,k}
}
+
i \int_0^t
e^{i(\lambda_{K,j}-E_{K,k})s}\,ds\,
\frac{\lambda_{K,j}-E_{K,j}}{E_{K,j}-E_{K,k}}.
$$
Consequently,
\begin{equation} \label{eq:Rneq0}
\begin{aligned}
R_{\neq 0}(t)
=
\frac{8}{L}
\sum_{j\in\mathcal J_{K,L}\setminus\{L/2\}}
|f_{K,j}(1)|^2
\Bigg|
&-i\,
\1_{|K|\ \mathrm{even}}
\1_{j\in\mathcal K_K\setminus\{0\}}\,
c_{K,j}\Gamma_{K,j}
\int_0^t
e^{ i(\lambda_{K,j}-E_{K,j})s}\,ds
\\
&+
\sum_{k\in\mathcal K_K\setminus\{0,j\}}
\left(
1-e^{i(\lambda_{K,j}-E_{K,k})t}
\right)
\frac{c_{K,k}\Gamma_{K,k}}
{E_{K,j}-E_{K,k}}
\\
&+
i(\lambda_{K,j}-E_{K,j})
\sum_{k\in\mathcal K_K\setminus\{0,j\}}
\int_0^t
e^{\mathrm i(\lambda_{K,j}-E_{K,k})s}\,ds\,
\frac{c_{K,k}\Gamma_{K,k}}
{E_{K,j}-E_{K,k}}
\Bigg|^2.
\end{aligned}
\end{equation}
We introduce the operators $T$ and $T'$ given by
\begin{equation}
    \label{T:action}
 \begin{split}
    (T g)(j) &= \sum_{k\in \mathcal{K}_K\setminus\{0,j\}} \frac{E'_{K,k}}{E_{K,j} - E_{K,k}}g(k) = \sum_{k\in \mathcal{K}_K\setminus\{0,j\}} \frac{\frac{2\pi}{L}\sin\left(\tfrac{2\pi k}{L}\right)}{{\cos\left(\tfrac{2\pi k}{L}\right)-\cos\left(\tfrac{2\pi j}{L}\right)}} g(k),\\
    ( T' g)(j) &= \sum_{k\in \mathcal{K}_K\setminus\{0,j\}  }\frac{E'_{K,j}}{E_{K,j} - E_{K,k}}g(k)=
    \sum_{k\in \mathcal{K}_K\setminus\{0,j\}} \frac{\frac{2\pi}{L}\sin\left(\tfrac{2\pi j}{L}\right)}{{\cos\left(\tfrac{2\pi k}{L}\right)-\cos\left(\tfrac{2\pi j}{L}\right)}} g(k),
\end{split}
\end{equation}
where the functions $g_K^t$ and $h_K^t$ read
\begin{equation}
\label{eq:def_g_h}
g_K^t(k) = \frac{c_{K,k}\Gamma_{K,k}}{E'_{K,k}} e^{-iE_{K,k}t}, \qquad
h_K^t(k) = c_{K,k}\Gamma_{K,k} e^{-i E_{K,k}t}
\end{equation}
for
\begin{displaymath}
E'_{K,k} = \frac{4\pi}{L}\cos\left(\frac{\pi K}{L}\right) \sin\left(\frac{2\pi k}{L}\right).
\end{displaymath}
With these definitions, we have 
$$
\begin{aligned}
\sum_{k\in\mathcal K_K\setminus\{0,j\}}
\left(
1-e^{i(\lambda_{K,j}-E_{K,k})t}
\right)
\frac{c_{K,k}\Gamma_{K,k}}
{E_{K,j}-E_{K,k}}
=
(Tg_K^0)(j)
-e^{i\lambda_{K,j}t}(Tg_K^t)(j),
\end{aligned}
$$
and
$$
(\lambda_{K,j}-E_{K,j})
\sum_{k\in\mathcal K_K\setminus\{0,j\}}
\int_0^t
e^{ i(\lambda_{K,j}-E_{K,k})s}\,ds\,
\frac{c_{K,k}\Gamma_{K,k}}
{E_{K,j}-E_{K,k}}
=
\frac{\lambda_{K,j}-E_{K,j}}{E'_{K,j}}
\int_0^t
e^{i\lambda_{K,j}s}
(T'h_K^s)(j)\,ds.
$$
Therefore,
$$
\begin{aligned}
R_{\neq 0}(t)
=
\frac{8}{L}
\sum_{j\in\mathcal J_{K,L}\setminus\{L/2\}}
|f_{K,j}(1)|^2
\Bigg|
&- i\,
\1_{|K|\ \mathrm{even}}
\1_{j\in\mathcal K_K\setminus\{0\}}
c_{K,j}\Gamma_{K,j}
\int_0^t
e^{i(\lambda_{K,j}-E_{K,j})s}\,ds
\\
&+(Tg_K^0)(j)
-e^{i\lambda_{K,j}t}(Tg_K^t)(j)
\\
&+i\,
\frac{\lambda_{K,j}-E_{K,j}}{E'_{K,j}}
\int_0^t
e^{i\lambda_{K,j}s}
(T'h_K^s)(j)\,ds
\Bigg|^2.
\end{aligned}
$$

Here the exclusion \(j\neq L/2\) ensures that \(E'_{K,j}\neq 0\), while the exclusion \(k\neq0\) ensures that \(g_K^t(k)\) is well-defined.

We now estimate $R_{\neq 0}(t)$. Since $K$ is fixed, for all sufficiently large $L$, Proposition \ref{prop:3:8:a} and the eigenvalue estimate in Proposition \ref{prop:eigenvalue_comparison_spin} imply
$$
|f_{K,j}(1)|^2\leq \frac{C}{L},
\qquad
\left|
\frac{\lambda_{K,j}-E_{K,j}}{E'_{K,j}}
\right|
\leq C,
\qquad
j\in\mathcal J_{K,L}\setminus\{L/2\},
$$
with a constant $C$ independent of $j$ and $L$. Using
$$
|a_1+a_2+a_3+a_4|^2
\leq
4\sum_{\ell=1}^4|a_\ell|^2
$$
in the expression for $R_{\neq 0}(t)$, we obtain
\begin{equation}\label{R_neq_0_terms}
\begin{aligned}
R_{\neq 0}(t)
\leq &
\frac{C}{L^2}
\sum_{j\in\mathcal J_{K,L}\setminus\{L/2\}}
\1_{|K|\ \mathrm{even}}
\1_{j\in\mathcal K_K\setminus\{0\}}
|c_{K,j}\Gamma_{K,j}|^2
\left|
\int_0^t
e^{i(\lambda_{K,j}-E_{K,j})s}\,ds
\right|^2
\\
&+
\frac{C}{L^2}
\sum_{j\in\mathcal J_{K,L}\setminus\{L/2\}}
|(Tg_K^0)(j)|^2+ \frac{C}{L^2}
\sum_{j\in\mathcal J_{K,L}\setminus\{L/2\}}
|(Tg_K^t)(j)|^2
\\
&+
\frac{C}{L^2}
\sum_{j\in\mathcal J_{K,L}\setminus\{L/2\}}
\left|
\int_0^t
e^{  i\lambda_{K,j}s}
(T'h_K^s)(j)\,ds
\right|^2.
\end{aligned}
\end{equation}
The first term satisfies
\begin{equation}\label{eq:firstRneq0bound}
\frac{C}{L^2}
\sum_{j\in\mathcal J_{K,L}\setminus\{L/2\}}
\1_{|K|\ \mathrm{even}}
\1_{j\in\mathcal K_K\setminus\{0\}}
|c_{K,j}\Gamma_{K,j}|^2
\left|
\int_0^t
e^{i(\lambda_{K,j}-E_{K,j})s}\,ds
\right|^2
\leq
\frac{Ct^2}{L^2}
\sum_{k\in\mathcal K_K\setminus\{0\}}
|c_{K,k}\Gamma_{K,k}|^2 \leq \frac{Ct^2 \varepsilon^4}{L^2}
\end{equation}
where in the last step we used the assumption \eqref{c_assumption_1}.

For the last term on the RHS of \eqref{R_neq_0_terms}, the Cauchy--Schwarz inequality in the time variable gives
\begin{equation} \label{eq:Rneq0lasttermbound}
\begin{aligned}
\sum_{j\in\mathcal J_{K,L}\setminus\{L/2\}}
\left|
\int_0^t
e^{ i\lambda_{K,j}s}
(T'h_K^s)(j)\,ds
\right|^2 & \leq
t\int_0^t
\sum_{j\in\mathcal J_{K,L}\setminus\{L/2\}}
|(T'h_K^s)(j)|^2\,ds \\
& \leq
t^2
\sup_{0\leq s\leq t}
\sum_{j\in\mathcal J_{K,L}\setminus\{L/2\}}
|(T'h_K^s)(j)|^2.
\end{aligned}
\end{equation}

Consequently,
\begin{equation} \label{eq:Rneq0midbound}
\begin{aligned}
R_{\neq 0}(t)
\leq{}&
\frac{Ct^2\varepsilon^4}{L^2}+ 
\frac{C}{L^2}
\sum_{j\in\mathcal J_{K,L}\setminus\{L/2\}}
|(Tg_K^0)(j)|^2
+
\frac{C}{L^2}
\sum_{j\in\mathcal J_{K,L}\setminus\{L/2\}}
|(Tg_K^t)(j)|^2
\\
&+
\frac{Ct^2}{L^2}
\sup_{0\leq s\leq t}
\sum_{j\in\mathcal J_{K,L}\setminus\{L/2\}}
|(T'h_K^s)(j)|^2.
\end{aligned}
\end{equation}
By Theorem \ref{thm:uniformT}, the operators
$$
T,T':
\ell^2(\mathcal K_K\setminus\{0\})
\longrightarrow
\ell^2(\mathcal J_{K,L}\setminus\{L/2\})
$$
have operator norms bounded uniformly in $L$. Therefore
$$
\sum_{j\in\mathcal J_{K,L}\setminus\{L/2\}}
|(Tg_K^t)(j)|^2
\leq
C\sum_{k\in\mathcal K_K\setminus\{0\}}
|g_K^t(k)|^2
=
C\sum_{k\in\mathcal K_K\setminus\{0\}}
\left|
\frac{c_{K,k}\Gamma_{K,k}}{E'_{K,k}}
\right|^2,
$$
and likewise
$$
\sum_{j\in\mathcal J_{K,L}\setminus\{L/2\}}
|(T'h_K^s)(j)|^2
\leq
C\sum_{k\in\mathcal K_K\setminus\{0\}}
|h_K^s(k)|^2
=
C\sum_{k\in\mathcal K_K\setminus\{0\}}
|c_{K,k}\Gamma_{K,k}|^2\leq C \varepsilon^4.
$$
Inserting these estimates into \eqref{eq:Rneq0midbound}, we find
$$
R_{\neq 0}(t)
\leq
\frac{C}{L^2}
\sum_{k\in\mathcal K_K\setminus\{0\}}
\left|
\frac{c_{K,k}\Gamma_{K,k}}{E'_{K,k}}
\right|^2
+
\frac{Ct^2}{L^2}
\varepsilon^4.
$$
Finally, since
$$
E'_{K,k}
=
\frac{1}{L}\widetilde E'_{K,k},
$$
assumptions \eqref{c_assumption_1} imply
$$
\frac{1}{L^2}
\sum_{k\in\mathcal K_K\setminus\{0\}}
\left|
\frac{c_{K,k}\Gamma_{K,k}}{E'_{K,k}}
\right|^2
=
\sum_{k\in\mathcal K_K\setminus\{0\}}
\left|
\frac{c_{K,k}\Gamma_{K,k}}
{\widetilde E'_{K,k}}
\right|^2
\leq \varepsilon^2.
$$
We conclude that
\begin{equation} \label{eq:Rneq0finalbound}
R_{\neq 0}(t)
\leq
C\varepsilon^2
+
C\frac{\varepsilon^4t^2}{L^2}.
\end{equation}
Combining \eqref{norm_difsec:main-krooff_2}, \eqref{eq:R0boundfinal}, \eqref{eq:RL/2finalbound} and \eqref{eq:Rneq0finalbound} we obtain
\begin{equation}
\label{norm_diff_final}
\left\| \Psi_K^t - \Psi_K^{t,\mathrm{free}} \right\|_{\mathcal{H}_\Lambda}^2  \leq C\frac{\varepsilon^4t^2}{L^2}+C\frac{t^2}{L^5}+C\varepsilon^2.
\end{equation}
For fixed $\varepsilon$ and $L$ large enough the second term is smaller than the first one for any $t$. This finishes the proof of \eqref{theorem: main equation}.
\end{proof}

The rest of the section is devoted to a discussion related to Remarks \ref{remarks_theorem} (c), \ref{remarks_theorem} (d), and Section~\ref{sec:heuristics}. We start with the former one, in which we show an example of coefficients $c_{K,k}$ which satisfy the assumptions of the main theorem.

\begin{remark}\label{rem_end_of_proof}
Let \(K\in\mathbb Z\) be fixed, and fix 
$0<\varepsilon_0<1/3$. We verify that the conditions
\eqref{c_assumption_1} are satisfied, up to constants, by every
normalized family $(c_{K,k})_{k\in\mathcal K_K}$ such that
$c_{K,k}=0$ whenever  $k/L>\varepsilon$ with 
 $0<\varepsilon\leq\varepsilon_0$. We also assume that $L$
is sufficiently large, depending on $K$ and $\varepsilon$, so
that $|K|/L\leq\varepsilon$. Let
$$
    I_\varepsilon
    :=
    \left\{
        k\in\mathcal K_K:
        c_{K,k}\neq0
    \right\},
    \qquad
    M_\varepsilon:=|I_\varepsilon|.$$
We have
$$
    M_\varepsilon\leq \varepsilon L+1.
$$
Set
$$
    S_K:=\sum_{k\in\mathcal K_K}|c_{K,k}|^2.
$$
It follows from \eqref{eq:normalization} and the
Cauchy--Schwarz inequality that
\begin{align*}
    1
   \geq
    S_K-\frac{2M_\varepsilon}{L}S_K
    =
    \left(1-\frac{2M_\varepsilon}{L}\right)S_K,
\end{align*}
and thus, for $L$ sufficiently large,
we obtain
\begin{equation}
\label{eq:coefficient_l2_bound}
    \sum_{k\in\mathcal K_K}|c_{K,k}|^2\leq C,
\end{equation}
where $C$ is independent of $L$ and $\varepsilon$.

For $k\in I_\varepsilon$, we have
\begin{align*}
    |\Gamma_{K,k}|=
    \left|
        \cos\left(\frac{\pi K}{L}\right)
        -
        \cos\left(\frac{2\pi k}{L}\right)
    \right|
    =
    2\left|
        \sin\left(\frac{\pi(K+2k)}{2L}\right)
        \sin\left(\frac{\pi(K-2k)}{2L}\right)
    \right|
    \leq C\varepsilon^2,
\end{align*}
where we used $k/L\leq\varepsilon$ and
$|K|/L\leq\varepsilon$. Consequently,
\begin{equation}
\label{eq:first_packet_condition_verification}
    \sum_{k\in\mathcal K_K}
        |c_{K,k}\Gamma_{K,k}|^2
    \leq
    C\varepsilon^4
    \sum_{k\in\mathcal K_K}|c_{K,k}|^2
    \leq C\varepsilon^4.
\end{equation}

We next estimate the second quantity in
\eqref{c_assumption_1}. Recall that
$$
    E'_{K,k}
    =
    \frac{4\pi}{L}
    \cos\left(\frac{\pi K}{L}\right)
    \sin\left(\frac{2\pi k}{L}\right)
$$
and
$$
    \widetilde E'_{K,k}=L E'_{K,k}.
$$

Suppose first that $K=0$. Since $k/L\leq\varepsilon_0<1/3$,
we have
$\sin(2\pi k/L)\geq\sin(\pi k/L)$, and hence
\begin{align*}
    \left|
        \frac{\Gamma_{0,k}}{E'_{0,k}}
    \right|
    =
    \frac{
        2\sin^2\left(\frac{\pi k}{L}\right)
    }{
        \frac{4\pi}{L}
        \sin\left(\frac{2\pi k}{L}\right)
    }\leq
    CL\sin\left(\frac{\pi k}{L}\right)
    \leq C\varepsilon L.
\end{align*}

Now suppose that $K\neq0$. Put
$$
    a:=\frac{\pi K}{2L},
    \qquad
    x:=\frac{\pi k}{L}.
$$
Since $k\in\mathcal K_K\setminus\{0\}$, we have $k\geq1/2$.
Using $\sin(2x)\geq\sin x$, we obtain
\begin{align*}
    \frac{|\sin(a+x)|}{\sin(2x)}
    \leq
    \frac{|\sin(a+x)|}{\sin x}
    \leq
    |\sin a|\cot x+|\cos a|
    \leq
    \frac{\pi|K|}{2L}\frac{L}{\pi k}+1
    \leq C_K.
\end{align*}
Therefore,
\begin{align*}
    \left|
        \frac{\Gamma_{K,k}}{E'_{K,k}}
    \right|
    =
    \frac{
        2
        \left|
            \sin\left(\frac{\pi(K+2k)}{2L}\right)
            \sin\left(\frac{\pi(K-2k)}{2L}\right)
        \right|
    }{
        \frac{4\pi}{L}
        \cos\left(\frac{\pi K}{L}\right)
        \sin\left(\frac{2\pi k}{L}\right)
    }\leq C_K\varepsilon L,
\end{align*}
where we used
$$
    \left|
        \sin\left(\frac{\pi(K-2k)}{2L}\right)
    \right|
    \leq C\varepsilon
$$
and the fact that
$\cos(\pi K/L)\geq1/2$ for fixed $K$ and sufficiently large
$L$. Thus, for every fixed $K$,
$$
    \left|
        \frac{\Gamma_{K,k}}{\widetilde E'_{K,k}}
    \right|
    =
    \frac1L
    \left|
        \frac{\Gamma_{K,k}}{E'_{K,k}}
    \right|
    \leq C_K\varepsilon.
$$
Together with \eqref{eq:coefficient_l2_bound}, this gives
\begin{equation*}
\label{eq:second_packet_condition_verification}
    \sum_{k\in\mathcal K_K\setminus\{0\}}
    \left|
        \frac{c_{K,k}\Gamma_{K,k}}
             {\widetilde E'_{K,k}}
    \right|^2
    \leq
    C_K\varepsilon^2
    \sum_{k\in\mathcal K_K}|c_{K,k}|^2
    \leq C_K\varepsilon^2.
\end{equation*}
This yields \eqref{c_assumption_1}.
\end{remark}

Next we state a proposition that establishes a bound on the norm of the difference between the free and full evolution that is valid for short times, as discussed in Remark~\ref{remarks_theorem} (d).

\begin{proposition}\label{prop:shorttimes}
   Under the assumptions of Theorem \ref{thm:wave-kacket}, we have
\begin{equation}\label{norm_diff_short}
\|\Psi_K^t-\Psi_K^{t,\mathrm{free}}\|^2
\leq 2t^2\varepsilon^4.
\end{equation}
\end{proposition}

\begin{proof} Recall the correspondence between relative wave functions and two-spin states used in Section \ref{section:one-body-reduction} . If $f\in\mathcal H_{K,L}^{(1)}$, then we can extend $f$ to the integers as in \eqref{special_case_psi_K^0},  i.e.
$$
\mathcal U_K f
:=
\sum_{\substack{x,y\in\Lambda\\x\ne y}}
e^{\frac{\pi iK}{L}(x+y)}
f(x-y)\,S_x^+S_y^+\Omega.
$$
As before, we have
\begin{equation}\label{eq:norm_vector_Uk}
\|\mathcal U_K f\|^2
=2\sum_{\substack{x,y\in\Lambda\\x\ne y}}|f(x-y)|^2 =2L\sum_{X=1}^{L-1}|f(X)|^2.
\end{equation}

In the second equality, we used the fact that each separation \(X\) occurs for exactly \(L\) ordered pairs \((x,y)\). Equations \eqref{eq:Schrodinger_wave_fct} and \eqref{h_lambda_K_def} also give $$
H_\Lambda\mathcal U_K f=\mathcal U_K(h_{\Lambda,K}f).
$$
Set \(w_k(X)=\cos(2\pi kX/L)\). Then
$$
\Psi_{K,k}^0=\mathcal U_K(w_k/L).
$$
From \eqref{eq_onebody_Ham} we obtain that
 
$$(h_{\Lambda,K}w_k)(X)=E_{K,k}w_k(X) \qquad \text{for} \qquad 2\leq X\leq L-2.$$
At \(X=1\) we get
$$
((h_{\Lambda,K}-E_{K,k})w_k)(1)=\Gamma_{K,k}.
$$
The symmetry \(w_k(L-X)=(-1)^K w_k(X)\) gives the value
\((-1)^K\Gamma_{K,k}\) at \(X=L-1\). Hence

$$
(h_{\Lambda,K}-E_{K,k})w_k
=\Gamma_{K,k}
\bigl(\delta_{X,1}+(-1)^K\delta_{X,L-1}\bigr).
$$

Define a relative wave function and its corresponding spin state by

$$
b_K(X):=\frac{\delta_{X,1}+(-1)^K\delta_{X,L-1}}{L},
\qquad B_K:=\mathcal U_K b_K.
$$

The function $b_K$ satisfies the required symmetry conditions. We arrive at
\begin{equation}\label{eq:hamminusE}
(H_\Lambda-E_{K,k})\Psi_{K,k}^0
=\mathcal U_K\!\left(
\frac{(h_{\Lambda,K}-E_{K,k})w_k}{L}
\right)=\Gamma_{K,k}\mathcal U_K b_K
=\Gamma_{K,k}B_K.
\end{equation}
Notice that $B_K$ is independent of $k$.  We also need its norm. The only nonzero values of \(b_K\) on \(\{1,\ldots,L-1\}\) are \(b_K(1)=1/L\) and \(b_K(L-1)=(-1)^K/L\). Thus, by \eqref{eq:norm_vector_Uk}, we have
$$
\|B_K\|^2
=2L\left(|b_K(1)|^2+|b_K(L-1)|^2\right)
=2L\left(\frac1{L^2}+\frac1{L^2}\right)
=\frac4L.
$$
Now, let us write $u(t)=\Psi_K^t$ and $v(t)=\Psi_K^{t,\mathrm{free}}$. From \eqref{psi_K^t-intro},
$$
v(t)=\sum_{k\in\mathcal K_K}
c_{K,k}e^{-iE_{K,k}t}\Psi_{K,k}^0.
$$
Thus
$$
i\partial_t v(t)
=\sum_{k\in\mathcal K_K}
c_{K,k}E_{K,k}e^{-iE_{K,k}t}\Psi_{K,k}^0.
$$
Subtracting this expression from $H_\Lambda v(t)$ and using \eqref{eq:hamminusE}, we find
$$
(H_\Lambda-i\partial_t)v(t)
=\sum_{k\in\mathcal K_K}
c_{K,k}e^{-iE_{K,k}t}
(H_\Lambda-E_{K,k})\Psi_{K,k}^0=B_K\sum_{k\in\mathcal K_K}
c_{K,k}\Gamma_{K,k}e^{-iE_{K,k}t}.
$$
Now, set $D(t)=u(t)-v(t)$ and
$R(t)=(H_\Lambda-i\partial_t)v(t)$. Since $i\partial_tu=H_\Lambda u$ and $u(0)=v(0)$, we have
$$
i\partial_tD=H_\Lambda D+R,\qquad D(0)=0.
$$
By the Duhamel formula we obtain
$$
D(t)=-i\int_0^t e^{-iH_\Lambda(t-s)}R(s)\,ds
$$
which yields
$$
D(t)
=-i\int_0^t e^{-iH_\Lambda(t-s)}B_K
\left(\sum_{k\in\mathcal K_K}
c_{K,k}\Gamma_{K,k}e^{-iE_{K,k}s}\right)ds.
$$

Unitarity of \(e^{-iH_\Lambda(t-s)}\), the triangle inequality for the integral, and \(\|B_K\|=2/\sqrt L\) prove that
\begin{equation}\label{shorttime1}
\|\Psi_K^t-\Psi_K^{t,\mathrm{free}}\|
\leq
\frac{2}{\sqrt L}\int_0^t
\left|\sum_{k\in\mathcal K_K}
c_{K,k}\Gamma_{K,k}e^{-iE_{K,k}s}\right|\,ds.
\end{equation}
The desired estimate then follows from Cauchy–Schwarz and the first assumption in \eqref{c_assumption_1}. Indeed, we have
$$
\left|\sum_{k\in\mathcal K_K}
c_{K,k}\Gamma_{K,k}e^{-iE_{K,k}s}\right|^2
\leq |\mathcal K_K|
\sum_{k\in\mathcal K_K}|c_{K,k}\Gamma_{K,k}|^2\leq \frac L2\,\varepsilon^4,
$$
where $|\mathcal K_K|\leq L/2$ follows from \eqref{def:K_K}. Inserting this into \eqref{shorttime1} yields
$\|D(t)\|\leq\sqrt2\,t\varepsilon^2$, which ends the proof.
\end{proof}

Lastly, we prove the claims from the discussion in Section~\ref{sec:heuristics}.

\begin{remark}\label{remark_asymptotic_bound}
Using the notation and results from Section~\ref{section:spectral problem} and the proof of Theorem~\ref{thm:wave-kacket}, we can derive an exact expression for the quantity $\delta_k$ from Section~\ref{sec:heuristics}. Recall that there we consider the two-spin wave state $\Psi^{(\textbf{k})} = N  S^+(\textbf{k})S^+(-\textbf{k})\vac$ with $K=0$ and associated energy $E_{0,k} = 2\big(1-\cos(2\pi k /L)\big)$. We then decompose
\begin{equation}
\Psi^{(\textbf{k})} = \sum_\lambda a_\lambda \widetilde{\Psi}_\lambda = a_{\lambda_k} \widetilde{\Psi}_{\lambda_k} + \sum_{\lambda \neq \lambda_k} a_\lambda \widetilde{\Psi}_\lambda, 
\end{equation}
where $\widetilde{\Psi}_\lambda$ is a normalized eigenstate of $H_\Lambda$ (restricted to $K=0$) with eigenvalue $\lambda$, $\lambda_k$ is the eigenvalue closest to $E_{0,k}$ (the $\widetilde{\Psi}_{\lambda_k}$ is the state $\widetilde{\Psi}_{\textbf{k}}$ from \eqref{def:Psi}), and $a_{\lambda_n} = \scp{\widetilde{\Psi}_{\lambda_n}}{\Psi^{(\textbf{k})}}$. With that notation, we aim at computing the quantity
\begin{equation}
\delta_k^2 = \sum_{\lambda \neq \lambda_k} |a_{\lambda}|^2 = 1 - |a_{\lambda_k}|^2.
\end{equation}
From the direct computation \eqref{exact_f_and_cos_overlap}, we find
\begin{equation}
a_{\lambda_n} = \sqrt{2L} \scp{f_{0,n}}{f_{\mathrm{free},k}} = \sqrt{2L} \frac{2f_{0,n}(1)\Gamma_{0,k}}{\lambda_n - E_{0,k}} \frac{c_{0,k}}{L},
\end{equation}
where we have used the normalization $c_{0,k} = (1-\frac{2}{L})^{-1/2}$, and $f_{\mathrm{free},k}(x) = \frac{c_{0,k}}{L} \cos(\frac{2\pi k}{L}x) = \frac{c_{0,k}}{L} w_k(x)$. For $K=0$, the discussion of Section~\ref{section:spectral problem} simplifies. From the proof of Theorem~\ref{thm:spectral_decomposition}, we find the exact eigenvalues $\lambda_n= 2(1-\cos(n\pi/(L-1)))$ and eigenfunctions $f_n(x) = \cos(\frac{n \pi}{L-1}(x-\frac{1}{2}))$ for $n=0,\ldots,L-2$ (there is no $\sinh$ solution). Since $K=0$ we only consider $n$ even as in Corollary~\ref{cor:lambdakj}, and the relative wave function corresponding to the normalized physical eigenstate is $f_{0,k+1}(x) = (2/(L-1))^{1/2}\cos(\frac{2k \pi}{L-1}(x-\frac{1}{2}))$. Hence, we have
\begin{equation}\label{exact_energy_diff_one_dim_K_zero}
E_{0,k} - \lambda_{0,k+1} = 2 \left( \cos(2\pi k/(L-1)) - \cos(2\pi k/L) \right) = -4 \sin(p_{k,L} + \varepsilon_{k,L})\sin(\varepsilon_{k,L}),
\end{equation} 
with $p_{k,L} = \frac{2\pi k}{L}$ and $\varepsilon_{k,L} = \frac{p_{k,L}}{2(L-1)}$. (Note in particular that this implies the asymptotic behavior claimed in \eqref{Delta_E} for this special setting.) Furthermore,
\begin{align}\label{exact_delta_1_d_details}
\delta_k^2 &= 1 - a_{\lambda_{0,k+1}}^2 \nonumber\\
&= 1 - \left( \frac{2\sqrt{2}}{\sqrt{L}} c_{0,k} \frac{f_{0,n}(1)\Gamma_{0,k}}{\lambda_n - E_{0,k}} \right)^2 \nonumber\\
&= 1 - \frac{1}{(L-1)(L-2)} \frac{\cos^2\big(\frac{k\pi}{L-1}\big) \Big(1-\cos\big(\frac{2\pi k}{L}\big)\Big)^2}{\sin^2(p_{k,L} + \varepsilon_{k,L}) \sin^2(\varepsilon_{k,L})} \nonumber\\
&= 1 - \frac{\sin^2(p_{k,L}/2)}{(L-1)(L-2)} \left( \frac{1}{\sin(\varepsilon_{k,L})} - \frac{1}{\sin(p_{k,L} + \varepsilon_{k,L})} \right)^2,
\end{align}
where the last step follows from standard trigonometric identities. Evaluating the asymptotic behavior by Taylor expanding the sine functions, we arrive at $\delta_k^2 \approx \frac{p^2_{k,L}}{12}$ (in the $k \sim \varepsilon L$ regime), as claimed in \eqref{delta_approx_heuristic_argument}.
\end{remark}

\appendix

\section{Uniform boundedness of the singular operators}
\label{appendix:boundedness}

We now show that the operators $T$ and $T'$ from \eqref{T:action} are uniformly bounded. More precisely, we will prove the following

 \begin{theorem}[Uniform Boundedness of $T$ and $T'$]
\label{thm:uniformT}
Let $L \in \mathbb{N}$ ($L \ge 2$) and $K \in \mathbb{Z}$.  Let $\mathcal{J}_{K,L}$ and $\mathcal{K}_K$ be defined as in Corollary \ref{cor:lambdakj} and \eqref{def:K_K}, respectively. The operators $T, T': \ell^2(\mathcal{K}_K \setminus \{0\}) \to \ell^2(\mathcal{J}_{K,L} \setminus \{L/2\})$ defined by
\begin{align}
    (Tg)(j) &= \sum_{k \in \mathcal{K}_K \setminus \{0, j\}} \frac{\frac{2\pi}{L} \sin\left(\frac{2\pi k}{L}\right)}{\cos\left(\frac{2\pi k}{L}\right) - \cos\left(\frac{2\pi j}{L}\right)} g(k), \label{eq:def_T_gen} \\[6pt]
    (T'g)(j) &= \sum_{k \in \mathcal{K}_K \setminus \{0, j\}} \frac{\frac{2\pi}{L} \sin\left(\frac{2\pi j}{L}\right)}{\cos\left(\frac{2\pi k}{L}\right) - \cos\left(\frac{2\pi j}{L}\right)} g(k), \label{eq:def_T_prime_gen}
\end{align}
are bounded uniformly with respect to both $L$ and $K$. Specifically, for all $L \ge 2$ and all $K \in \mathbb{Z}$,
\begin{equation}
    \|T\|_{\ell^2 \to \ell^2} \le 2\sqrt{2}\pi \qquad \text{and} \qquad \|T'\|_{\ell^2 \to \ell^2} \le 2\sqrt{2}\pi.
\end{equation}
\end{theorem}

The proof of Theorem \ref{thm:uniformT} will use the boundedness of the Periodic Discrete Hilbert Transform $H_L$. It is defined on 
 $\ell^2(\Z_L)$, i.e. the space of $L$-periodic sequences, by   
\begin{equation}\label{eq:def_HL}
(H_L f)(j) = \frac{\pi}{L} \sum_{p=0, p \neq j}^{L-1} f(p) \cot\left(\frac{\pi(j-p)}{L}\right).
\end{equation}
Since $H_L$ is a convolution operator on the cyclic group $\Z_L$, it is fully translation-invariant. Its uniform boundedness follows from Plancherel's theorem.

\begin{lemma}\label{lem:plancherel}
The periodic discrete Hilbert transform $H_L: \ell^2(\Z_L) \to \ell^2(\Z_L)$ satisfies
$$ \|H_L f\|_{\ell^2(\Z_L)} \le \pi \|f\|_{\ell^2(\Z_L)}, $$
 i.e. $H_L$ is  bounded uniformly in $L$.
\end{lemma}
\begin{proof}
Let $c(j) = \frac{\pi}{L}\cot\left(\frac{\pi j}{L}\right)$ for $j \neq 0$ and $c(0) = 0$ be the kernel. The Fourier multipliers $m(k) = \widehat{c}(k)$ for $k \in \{0, \dots, L-1\}$ are given by the  discrete Fourier transform
$$ m(k) = \frac{\pi}{L} \sum_{j=1}^{L-1} \cot\left(\frac{\pi j}{L}\right) e^{-i \frac{2\pi j k}{L}}. $$
Evaluating this sum using standard trigonometric exponential series yields
\begin{equation}
m(k) = \begin{cases} 
0 & \text{if } k = 0, \\
-i\pi \left(1 - \frac{2k}{L}\right) & \text{if } 0 < k < L.
\end{cases}
\end{equation}
For all $0 < k < L$, the term $|1 - \frac{2k}{L}| < 1$. Thus, $\max_k |m(k)| < \pi$. By Plancherel's identity we get $\|H_L\|_{\ell^2 \to \ell^2} \le \pi$.
\end{proof}

We are ready to prove Theorem \ref{thm:uniformT}.

\begin{proof}[Proof of Theorem \ref{thm:uniformT}] \textit{Step 1.} We shall first prove the theorem when both $L$ and $|K|$ are even.
To this end we consider the operators $\tilde{T}$ and $\tilde{T}'$   acting on the finite-dimensional Hilbert space $\ell^2(0, \dots, L/2)$, defined for a vector $g, h \in \ell^2(0, \dots, L/2)$ at index $j \in \{0, \dots, L/2\}$ by:
\begin{equation}\label{eq:def_T}
(\tilde{T}g)(j) = \sum_{p=0, p\neq j}^{L/2} \frac{2\pi}{L} \frac{\sin\left(\frac{2\pi p}{L}\right)}{\cos\left(\frac{2\pi p}{L}\right) - \cos\left(\frac{2\pi j}{L}\right)} g(p),
\end{equation}
\begin{equation}\label{eq:def_T_tilde}
(\tilde{T}'h)(j) = \sum_{p=0, p\neq j}^{L/2} \frac{2\pi}{L} \frac{\sin\left(\frac{2\pi j}{L}\right)}{\cos\left(\frac{2\pi p}{L}\right) - \cos\left(\frac{2\pi j}{L}\right)} h(p).
\end{equation}
Obviously, proving uniform bounds for $\|\tilde{T}\|$ and  $\|\tilde{T}'\|$ implies the same for $\|T\|$ and $\|T'\|$ (as the former operators are extensions of the latter ones). \\
\noindent \textit{Step 2.} We use the trigonometric  identity:
\begin{equation}\label{eq:trig_identity}
\cot A + \cot B = \frac{\sin(A+B)}{\sin A \sin B}.
\end{equation}
Let $\theta_p = \frac{2\pi p}{L}$ and $\theta_j = \frac{2\pi j}{L}$. 
We get the following decomposition for $\tilde{T}$: setting  $A = \frac{\pi(j-p)}{L}$ and $B = \frac{\pi(j+p)}{L}$ in \eqref{eq:trig_identity}, we observe that $A+B = \theta_j$. The denominator becomes:
\[ \sin A \sin B = \sin\left(\frac{\theta_j - \theta_p}{2}\right) \sin\left(\frac{\theta_j + \theta_p}{2}\right) = \frac{1}{2}(\cos \theta_p - \cos \theta_j). \]
Thus, we obtain the identity:
\begin{equation}\label{eq:trig_tilde}
\frac{2\sin \theta_j}{\cos \theta_p - \cos \theta_j} = \cot\left(\frac{\pi(j-p)}{L}\right) + \cot\left(\frac{\pi(j+p)}{L}\right).
\end{equation}

For $T$ we set $A = \frac{\pi(p-j)}{L}$ and $B = \frac{\pi(p+j)}{L}$ in \eqref{eq:trig_identity}, we have $A+B = \theta_p$. The denominator yields $\frac{1}{2}(\cos \theta_j - \cos \theta_p) = -\frac{1}{2}(\cos \theta_p - \cos \theta_j)$. This yields:
\begin{equation}\label{eq:trig_T}
\frac{2\sin \theta_p}{\cos \theta_p - \cos \theta_j} = \cot\left(\frac{\pi(j-p)}{L}\right) - \cot\left(\frac{\pi(j+p)}{L}\right).
\end{equation}
We now bound $\tilde{T}$ and $\tilde{T}'$ by embedding their domains into $\Z_L$ via invariant subspaces.\\
\noindent \textit{Step 3.} We start the analysis by considering the operator $\tilde{T}$. Let $g \in \ell^2(0, \dots, L/2)$. We observe from \eqref{eq:def_T} that when $p=0$ or $p=L/2$, $\sin(\theta_p) = 0$, meaning the boundary elements lie in the null space of $\tilde{T}$. Thus, we assume without loss of generality that $g(0) = g(L/2) = 0$.

We define the odd periodic extension $f \in \ell^2(\Z_L)$ of $g$ as:
\begin{equation}\label{eq:oddperext}
f(p) = \begin{cases} 
g(p) & \text{if } p \in \{1, \dots, L/2 - 1\}, \\
-g(L-p) & \text{if } p \in \{L/2 + 1, \dots, L - 1\}, \\
0 & \text{if } p = 0 \text{ or } p = L/2.
\end{cases}
\end{equation}
Evaluating the translation-invariant operator $H_L$ on this extended function $f$ for an index $j \in \{0, \dots, L/2\}$ yields:
\begin{align*}
(H_L f)(j) &= \frac{\pi}{L} \sum_{p=1, p\neq j}^{L/2-1} g(p) \cot\left(\frac{\pi(j-p)}{L}\right) + \frac{\pi}{L} \sum_{p=L/2+1}^{L-1} (-g(L-p)) \cot\left(\frac{\pi(j-p)}{L}\right) \\
&= \frac{\pi}{L} \sum_{p=1, p\neq j}^{L/2-1} g(p) \cot\left(\frac{\pi(j-p)}{L}\right) - \frac{\pi}{L} \sum_{p=1}^{L/2-1}g(p)\cot\left(\frac{\pi(j+p)}{L}\right) .
\end{align*}
Applying \eqref{eq:trig_T}, for \(p\neq j\), we obtain
$$
(\tilde Tg)(j)=(H_Lf)(j)+\frac{\pi}{L}\cot\left(\frac{2\pi j}{L}\right)g(j)
$$
if $j \in \{1, \dots, L/2 - 1\}$, and $(\tilde Tg)(j)=(H_Lf)(j)$ for $j = 0$ and $j = L/2$. Computing the energy norm of the extension we obtain
\begin{equation}\label{eq:norm_odd}
\|f\|_{\ell^2(\Z_L)}^2 = \sum_{p=1}^{L/2-1} |g(p)|^2 + \sum_{p=1}^{L/2-1} |{-g(p)}|^2 = 2 \sum_{p=1}^{L/2-1} |g(p)|^2 = 2 \|g\|_{\ell^2(0, \dots, L/2)}^2.
\end{equation}
It follows that $\|f\|_{\ell^2(\Z_L)} = \sqrt{2} \|g\|_{\ell^2}$. Using Lemma \ref{lem:plancherel} and the bound 
$$\biggl\lvert\frac{\pi}{L}\cot\left(\frac{2\pi j}{L}\right)\biggr\rvert\leq \frac12$$
for $1\leq j\leq \frac{L}{2} -1$, we find
\begin{equation}
\|\tilde{T}g\|_{\ell^2} \le \|H_L f\|_{\ell^2(\Z_L)}+\frac12 \|g\|_{\ell^2} \le \pi \|f\|_{\ell^2(\Z_L)} +\frac12 \|g\|_{\ell^2}= (\sqrt{2} \pi+\frac12) \|g\|_{\ell^2}.
\end{equation}
\noindent \textit{Step 4.} To deal with the operator $\tilde{T}'$ we consider an even extension. For $h \in \ell^2(0, \dots, L/2)$, the boundary points do not automatically vanish. We construct the even periodic extension with boundary doubling $f \in \ell^2(\Z_L)$ via
\begin{equation}\label{eq:evenperext}
f(p) = \begin{cases} 
h(p) & \text{if } p \in \{1, \dots, L/2 - 1\}, \\
h(L-p) & \text{if } p \in \{L/2 + 1, \dots, L - 1\}, \\
2h(0) & \text{if } p = 0, \\
2h(L/2) & \text{if } p = L/2.
\end{cases}
\end{equation}
Evaluating $(H_L f)(j)$ for $j \in \{1, \dots, L/2-1\}$ splits the domain into its symmetric parts:
\begin{align*}
(H_L f)(j) &= \frac{\pi}{L} \sum_{p=1, p\neq j}^{L/2-1} h(p)  \cot\left(\frac{\pi(j-p)}{L}\right) + \frac{\pi}{L} \sum_{p=1}^{L/2-1} h(p)\cot\left(\frac{\pi(j+p)}{L}\right) \\
&\quad + \frac{2\pi}{L} h(0)\cot\left(\frac{\pi j}{L}\right) + \frac{2\pi}{L} h(L/2)\cot\left(\frac{\pi(j-L/2)}{L}\right).
\end{align*}
Using \eqref{eq:trig_tilde}, we observe 
$$
(\tilde T'h)(j)=(H_Lf)(j)-\frac{\pi}{L}\cot\left(\frac{2\pi j }{L}\right)h(j)$$
 for all interior points, i.e. $j \in \{1,..., L/2 -1\}$ . At the boundaries $j \in \{0, L/2\}$, $(\tilde{T}'h)(j) = 0$ due to the $\sin(\theta_j)$ term in the numerator, providing a trivial bound. 
Computing the norm bound for the even extension we obtain
\begin{align*}
\|f\|_{\ell^2(\Z_L)}^2 &= |2h(0)|^2 + |2h(L/2)|^2 + 2\sum_{p=1}^{L/2-1} |h(p)|^2 \\
&\le 4\left(|h(0)|^2 + |h(L/2)|^2 + \sum_{p=1}^{L/2-1} |h(p)|^2\right) = 4 \|h\|_{\ell^2(0, \dots, L/2)}^2.
\end{align*}
Thus, $\|f\|_{\ell^2(\Z_L)} \le 2 \|h\|_{\ell^2}$. Using again Lemma \ref{lem:plancherel} we get
\begin{equation}
\|\tilde{T}'h\|_{\ell^2} \le \|H_L f\|_{\ell^2(\Z_L)}+\frac12\|h\|_{\ell^2} \le \pi  \|f\|_{\ell^2(\Z_L)}+\frac12\|h\|_{\ell^2} \le (2 \pi+\frac12) \|h\|_{\ell^2}.
\end{equation}
Thus, $\|\tilde{T}'\|_{\ell^2 \to \ell^2} \le 2 \pi + \frac12$ which ends the proof in the case $|K|$ and $L$ even.\\
\noindent \textit{Step 5.} In the case of even $|K|$ and odd $L$ there is no midpoint $L/2$. The proof goes through in the same way, except that in this case $\lvert\frac{\pi}{L}\cot(\frac{2\pi j}{L})\rvert\leq 1$, and we get the constant $2 \pi + 1$ instead of $2 \pi + \frac12$. \\
\noindent \textit{Step 6.} We shall now briefly sketch how to deal with the case of even $L$ and odd $|K|$ (the case for odd $L$ is a straightforward adaptation). We embed $g \in \ell^2(\mathcal{K}_K)$ into $\ell^2(\mathbb{Z}_{2L})$ by defining $f(m)$  on odd indices $m \in \{1, 3, \dots, 2L-1\}$. For $\tilde{T}$ we use an odd extension
$$f(m) = \begin{cases}   g(m/2) & \text{if } m \in \{1, 3, \dots, L-1\}, \\   -g((2L-m)/2) & \text{if } m \in \{L+1, L+3, \dots, 2L-1\}, \\   0 & \text{if } m \text{ is even.}   \end{cases}$$
For $\tilde{T}'$ we use an even extension
$$f(m) = \begin{cases}   g(m/2) & \text{if } m \in \{1, 3, \dots, L-1\} , \\   g((2L-m)/2) & \text{if } m \in \{L+1, L+3, \dots, 2L-1\} , \\   0 & \text{if } m \text{ is even.}   \end{cases}$$
We compute
$$(H_{2L} f)(2j) = \frac{\pi}{2L} \sum_{\substack{m=0 \\ m \neq 2j}}^{2L-1} f(m) \cot\left(\frac{\pi(2j-m)}{2L}\right) = \frac{\pi}{2L} \sum_{\substack{m=1 \\ m \text{ odd}}}^{2L-1} f(m) \cot\left(\frac{\pi(2j-m)}{2L}\right).$$

Splitting the summation into $m < L$ and $m > L$ and substituting $m' = 2L - m$ we get
$$\cot\left(\frac{\pi(2j-(2L-m'))}{2L}\right) = \cot\left(\frac{\pi(2j+m')}{2L} - \pi\right) = \cot\left(\frac{\pi(2j+m')}{2L}\right).$$
Combining both halves yields
$$(H_{2L} f)(2j) = \frac{\pi}{2L} \sum_{\substack{m=1 \\ m \text{ odd}}}^{L-1} g(m/2) \left[ \cot\left(\frac{\pi(2j-m)}{2L}\right) \mp \cot\left(\frac{\pi(2j+m)}{2L}\right) \right]$$
where the upper sign corresponds to the odd extension and the lower sign to the even extension.
Replacing $m = 2k$ and $2j = n$
$$\frac{\pi(2j-2k)}{2L} = \frac{\pi(j-k)}{L}, \qquad \frac{\pi(2j+2k)}{2L} = \frac{\pi(j+k)}{L}$$
and  applying trigonometric identities \eqref{eq:trig_tilde} and \eqref{eq:trig_T} shows that for every integer target index $j$ we get
$$(H_{2L} f)(2j) = \frac12(\tilde{T}g)(j) \qquad \text{and} \qquad (H_{2L} f)(2j) = \frac12(\tilde{T}'h)(j).$$
Applying Lemma \ref{lem:plancherel} ends the proof. 
\end{proof}


\begin{thebibliography}{}

\bibitem{Benedikter-2017}
{\sc N. Benedikter}, {\em Interaction Corrections to Spin Wave Theory in the Quantum Heisenberg Ferromagnet}, Math. Phys., Anal. Geom. 20, 1-21 (2017)

\bibitem{Bethe-31}
{\sc H. A. Bethe}, {\em Zur Theorie der Metalle. I. Eigenwerte und Eigenfunktionen der linearen Atomkette}, Zeit. für Physik 71, 205 (1931)

\bibitem{Blo-30}
{\sc F. Bloch}, {\em Zur Theorie des Ferromagnetismus}, Z. Physik 61, 206–219 (1930)

\bibitem{Blo-32}
{\sc F. Bloch}, {\em Zur Theorie des Austauschproblems und der Remanenzerscheinung der Ferromagnetika}, Z. Physik 74, 295–335 (1932)

\bibitem{Blochetal-12}
{\sc I. Bloch, J. Dalibard and S.  Nascimbène}, {\em Quantum simulations with ultracold quantum gases}. Nature Phys 8, 267-276 (2012).

\bibitem{Bose-03}
{\sc S. Bose}, {\em Quantum Communication through an Unmodulated Spin Chain}, Phys. Rev. Lett. 91 (20), 207901 (2003)

\bibitem{Bose-07}
{\sc S. Bose}, {\em Quantum communication through spin chain dynamics: an introductory overview}, Contemporary Physics, 48 (1), 13–30 (2007)
   
\bibitem{CorGiu-12}
{\sc M. Correggi, A. Giuliani}, {\em The free energy of the quantum Heisenberg ferromagnet at large spin}, J. Stat. Phys., 149 pp. 234–245 (2012) 

\bibitem{CorGiuSei-14}
{\sc M. Correggi, A. Giuliani, R. Seiringer}, {\em Validity of spin wave theory for the quantum Heisenberg model}, EPL (Europhysics Letters) 108 (2), 20003 (2014)

\bibitem{CorGiuSei-15}
{\sc M. Correggi, A. Giuliani, R. Seiringer}, {\em Validity of the spin wave approximation for the free energy of the Heisenberg ferromagnet}, Commun. Math. Phys. 339 (1), 279-307 (2015)

\bibitem{ConSol-91}
{\sc J. G. Conlon, J. P. Solovej}, {\em Upper bound on the free energy of the spin 1/2
Heisenberg ferromagnet}, Lett. Math. Phys., 23 pp. 223–231 (1991) 

\bibitem{Dys-56}
{\sc F. J. Dyson}, {\em General Theory of Spin Waves Interactions}, Phys. Rev. 102, 1217-1230 (1956)

\bibitem{Dys-56_2}
{\sc F. J. Dyson}, {\em Thermodynamic Behavior of an Ideal Ferromagnet}, Phys. Rev. 102, 1230-1244 (1956)

\bibitem{Fukuhara-13a}
{\sc T. Fukuhara et al.}, {\em Quantum dynamics of a mobile spin impurity}, Nature Physics 9, 235–241 (2013)

\bibitem{Fukuhara-13b}
{\sc T. Fukuhara et al.}, {\em Microscopic observation of magnon bound states and their dynamics}, Nature 502, 76–79 (2013)

\bibitem{HolPri-40}
{\sc T. Holstein, H. Primakoff}, {\em Field Dependence of the Intrinsic Domain Magnetization of a Ferromagnet}, Phys. Rev. 58, 1098-1113 (1940)

\bibitem{Klippel-26}
{\sc A. Klippel}, {\em Free-Energy Asymptotics of the Two-Dimensional Quantum Heisenberg Ferromagnet}, arXiv preprint arXiv:2608.25506

\bibitem{Kosevich-90}
{\sc A.M. Kosevich, B.A. Ivanov and A.S. Kovalev}, {\em Magnetic Solitons}.  Physics Reports 194 (3-4), 117-238 (1990). 

\bibitem{Mattis-81}
{\sc D.C. Mattis}, {\em  The Theory of Magnetism I: Statics and Dynamics.}
Springer Series in Solid-State Sciences, Vol. 17. Springer-Verlag, Berlin, 1981.

\bibitem{NapSei-19}
{\sc M. Napiórkowski, R. Seiringer}, {\em Free energy asymptotics of the quantum Heisenberg spin chain}, Lett. Math. Phys. 111, 31 (2021)

\bibitem{Simon-11}
{\sc J. Simon et al.}, {\em Quantum simulation of antiferromagnetic spin chains in an optical lattice}, Nature 472, 307–312 (2011)

\bibitem{Toth-93}
{\sc B. T\'oth}, {\em Improved Lower Bound on the Thermodynamic Pressure of the Spin 1/2 Heisenberg Ferromagnet}, Lett. Math. Phys., 28 pp. 75–84, (1993)

\bibitem{wortis}
{\sc M. Wortis}, {\em Bound States of Two Spin Waves in the Heisenberg Ferromagnet}, Phys. Rev. 132, 85 (1963).

\bibitem{WysChw-25}
{\sc P. Wysocki and J. Chwedeńczuk}, {\em Limits to Velocity of Signal Propagation in Many-Body Systems: A Quantum Information Perspective}, Phys. Rev. Lett.  134, 020201  (2025)

\end{thebibliography}
\end{document}